\documentclass[a4paper,UKenglish]{lipics-v2019}

\usepackage[dvipsnames,svgnames,x11names]{xcolor}
\usepackage{amssymb,amsmath,mathrsfs}
\usepackage{soul}
\usepackage{mathtools}
\usepackage{framed}
\usepackage{microtype}
\usepackage{listings}
\usepackage{mathrsfs}
\usepackage{amsfonts,amsthm,amsmath}
\usepackage{amstext}
\usepackage{graphicx,tikz}
\RequirePackage{fancyhdr}
\usepackage{boxedminipage}
\usepackage{todonotes}
\usepackage{graphics}
\usepackage{framed}
\usepackage{thm-restate}
\usepackage{thmtools}
\usepackage{enumerate}
\usepackage{float}
\usepackage{xspace}

\usetikzlibrary{decorations.pathmorphing}
\tikzset{snake it/.style={decorate, decoration=snake}}

\newcommand{\OO}{\mathcal{O}}
\newcommand{\cO}{\mathcal{O}}

\theoremstyle{plain}
\newtheorem{observation}[theorem]{Observation}

\mathchardef\myhyphen="2D

\newcommand{\yes}{\textsf{Yes}}
\newcommand{\no}{\textsf{No}}

\newcommand{\tw}{\mathtt{tw}}

\newcommand{\NP}{\textsf{NP}}

\newcommand{\ETH}{\textsf{ETH}}

\newcommand{\alg}{\textsf{ALG}}
\newcommand{\nil}{\ensuremath{\mathtt{nil}}}

\newcommand{\ma}{\ensuremath{\mathsf{Mark}}}
\newcommand{\probKPath}{\textsc{Long Path}\xspace}
\newcommand{\probKCycle}{\textsc{Long Cycle}\xspace}

\newcommand{\probFVS}{\textsc{Feedback Vertex Set}\xspace}

\newcommand{\probCycPacking}{\textsc{Cycle Packing}\xspace}

\title{ETH-Tight Algorithms for Long Path and Cycle on Unit Disk Graphs}
\titlerunning{ETH-Tight Algorithms for Long Path and Cycle on Unit Disk Graphs}

\author{Fedor V. Fomin}{University of Bergen, Bergen, Norway, fomin@ii.uib.no}{}{}{Research Council of Norway via the project MULTIVAL.}

\author{Daniel Lokshtanov}{University of California, Santa Barbara, USA, daniello@ucsb.edu}{}{}{ European Research Council (ERC) under the European Union’s Horizon
2020 research and innovation programme (grant no. 715744), and United States - Israel Binational Science Foundation grant no. 2018302.}

\author{Fahad Panolan}{Department of Computer Science and Engineering, IIT Hyderabad, Hyderbad, India, fahad@iith.ac.in}{}{}{}

\author{Saket Saurabh}{Department of Informatics, University of Bergen, Norway\\ The Institute of Mathematical Sciences, HBNI and IRL 2000 ReLaX, Chennai, India, saket@imsc.res.in}{}{}{European Research Council (ERC) under the European Union’s Horizon 2020 research and innovation programme (grant no. 819416), and Swarnajayanti Fellowship
grant DST/SJF/MSA-01/2017-18.}

\author{Meirav Zehavi}{Ben-Gurion University of the Negev, {Beer-Sheva, Israel}, meiravze@bgu.ac.il}{}{}{Israel Science Foundation grant no. 1176/18, and United States – Israel Binational Science Foundation grant no. 2018302.}

\authorrunning{F.V. Fomin, D. Lokshtanov, F. Panolan, S. Saurabh, M. Zehavi}

\Copyright{Fedor V. Fomin, Daniel Lokshtanov, Fahad Panolan, Saket Saurabh, Meirav Zehavi}

\ccsdesc[500]{Theory of computation~Fixed parameter tractability}
\ccsdesc[500]{Theory of computation~Computational geometry}

\keywords{Optimality Program, ETH, Unit Disk Graphs, Parameterized Complexity, Long Path, Long Cycle}

\EventEditors{Sergio Cabello and Danny Z. Chen}
\EventNoEds{2}
\EventLongTitle{36th International Symposium on Computational Geometry (SoCG 2020)}
\EventShortTitle{SoCG 2020}
\EventAcronym{SoCG}
\EventYear{2020}
\EventDate{June 23--26, 2020}
\EventLocation{Z\"{u}rich, Switzerland}
\EventLogo{socg-logo}
\SeriesVolume{164}

\begin{document}

\maketitle

\begin{abstract} 
We present an algorithm for the extensively studied {\sc Long Path} and {\sc Long Cycle} problems on unit disk graphs that runs in time $2^{\OO(\sqrt{k})}(n+m)$. Under the Exponential Time Hypothesis, {\sc Long Path} and {\sc Long Cycle} on unit disk graphs cannot be solved in time $2^{o(\sqrt{k})}(n+m)^{\OO(1)}$ [de Berg et al., STOC 2018], hence our algorithm is optimal. Besides the $2^{\OO(\sqrt{k})}(n+m)^{\OO(1)}$-time algorithm for the (arguably) much simpler {\sc Vertex Cover} problem by de Berg et al.~[STOC 2018] (which easily follows from the existence of a $2k$-vertex kernel for the problem), {\em this is the only known ETH-optimal fixed-parameter tractable algorithm on UDGs}. Previously, {\sc Long Path} and {\sc Long Cycle} on unit disk graphs were only known to be solvable in time $2^{\OO(\sqrt{k}\log k)}(n+m)$. This algorithm involved the introduction of a new type of a tree decomposition, entailing the design of a very tedious dynamic programming procedure. Our algorithm is substantially simpler: we completely avoid the use of this new type of tree decomposition. Instead, we use a marking procedure to reduce the problem to (a weighted version of) itself on a standard tree decomposition of width $\OO(\sqrt{k})$.
\end{abstract}

\section{Introduction}\label{sec:intro}

Unit disk graphs are the intersection graphs of disks of radius $1$ in the Euclidean plane. That is, given $n$ disks of radius $1$, we represent each disk by a vertex, and insert an edge  between two vertices if and only if their corresponding disks intersect.
Unit disk graphs form one of the best studied graph classes in computational geometry because of their use in modelling optimal facility location~\cite{wang1988study} and broadcast networks such as wireless, ad-hoc and sensor networks \cite{hale1980frequency,kammerlander1984c,yeh1984outage}. 
These applications have led to an extensive study of NP-complete problems on unit disk graphs in the realms of computational complexity and approximation algorithms. We refer the reader 
to~\cite{ClarkCJ90,DumitrescuP11,HuntMRRRS98} and the citations therein for these studies. However, these problems remain hitherto unexplored in the light of parameterized complexity with exceptions that are few and far between~\cite{alber2002geometric,Chan03,FominLS12,Jansen10,SmithW98}. 

We study the {\sc Long Path} (resp.~{\sc Long Cycle}) problem on unit disk graphs. Here, given a graph $G$ and an integer $k$, the objective is to decide whether $G$ contains a path (resp.~cycle) on at least $k$ vertices. To the best of our knowledge,  the {\sc Long Path} problem is among the five most extensively studied problems in Parameterized Complexity~\cite{DBLP:books/sp/CyganFKLMPPS15,DBLP:series/txcs/DowneyF13} (see Section \ref{sec:relatedWorks}). One of the best known open problems in Parameterized Complexity was to develop a $2^{\OO(k)}n^{\OO(1)}$-time algorithm for {\sc Long Path} on general graphs~\cite{DBLP:conf/coco/PapadimitriouY93}, that is, shaving the $\log k$ factor in the exponent of the previously best $2^{\OO(k
\log k)}n^{\OO(1)}$-time parameterized algorithm for this problem on general graphs~\cite{monien1985find}. This was resolved in the positive in the seminal work by Alon, Yuster and Zwick on color coding 25 years ago~\cite{AlonYZ}, which was recently awarded the IPEC-NERODE prize for the most outstanding research in Parameterized Complexity. In particular, the aforementioned $2^{\OO(k)}n^{\OO(1)}$-time algorithm for {\sc Long Path} on general graphs is optimal under the Exponential Time Hypothesis (\ETH).

Both {\sc Long Path} and {\sc Long Cycle} are known to be \NP-hard on unit disk graphs~\cite{ItaiPS82}, and cannot be solved in time  $2^{o(\sqrt{n})}$(and hence also in time $2^{o(\sqrt{k})}n^{\OO(1)}$) on unit disk graphs unless the \ETH\ fails~\cite{DBLP:conf/stoc/BergBKMZ18}.  Our contribution is an {\em optimal} parameterized algorithm for {\sc Long Path} (and~{\sc Long Cycle}) on unit disk graphs under the \ETH. Specifically, we prove the following theorem.

\begin{theorem}\label{thm:main}
{\sc Long Path} and {\sc Long Cycle} are solvable in time $2^{\OO(\sqrt{k})}(n+m)$ on unit disk graphs.
\end{theorem}

Two years ago, a celebrated work by de Berg et al.~\cite{DBLP:conf/stoc/BergBKMZ18} presented (non-parameterized) algorithms with running time $2^{\OO(\sqrt{n})}$ for a number of problems on intersection graphs of so called ``fat'', ``similarly-sized'' geometric objects, accompanied by matching lower bounds of $2^{\Omega(\sqrt{n})}$ under the ETH. Only for the {\sc Vertex Cover} problem this work implies an ETH-tight parameterized algorithm. More precisely, {\sc Vertex Cover} admits a $2k$-vertex kernel on general graphs~\cite{nemhauser1975vertex,DBLP:books/sp/CyganFKLMPPS15}, hence the $2^{\OO(\sqrt{n})}$-time algorithm for {\sc Vertex Cover} by de Berg et al.~\cite{DBLP:conf/stoc/BergBKMZ18} is trivially a $2^{\OO(\sqrt{k})}n^{\OO(1)}$-time fixed-parameter tractable algorithm for this problem. None of the other problems in \cite{DBLP:conf/stoc/BergBKMZ18} is known to admit a linear-vertex kernel, and we know of no other work that presents a $2^{\OO(\sqrt{k})}n^{\OO(1)}$-time algorithm for any basic problem on unit disk graphs. Thus, we present the second known ETH-tight fixed-parameter tractable algorithm for a basic problem on unit disk graphs, or, in fact, on any family of geometric intersection graphs of fat objects. In a sense, our work is the first time where tight \ETH-optimality of fixed-parameter tractable algorithms on unit disk graphs is explicitly answered. (The work of de Berg et al.~\cite{DBLP:conf/stoc/BergBKMZ18} primarily concerned non-parameterized algorithms.) We believe that our work will open a door to the realm to an ETH-tight optimality program for fixed-parameter tractable algorithms on intersection graphs of fat geometric objects.

Prior to our work, {\sc Long Path} and {\sc Long Cycle} were known to be solvable in time $2^{\OO(\sqrt{k}\log k)}(n+m)$ on unit disk graphs~\cite{DBLP:journals/dcg/FominLPSZ19}. Thus, we shave the $\log k$ factor in the exponent in the running time, and thereby, in particular, achieve optimality. Our algorithm is substantially simpler, both conceptually and technically, than the previous algorithm as we explain below. The main tool in the previous algorithm (of \cite{DBLP:journals/dcg/FominLPSZ19}) for {\sc Long Path} (and {\sc Long Cycle}) on unit disk graphs was a new (or rather refined) type of a tree decomposition.\footnote{We refer the reader to Section \ref{sec:prelims} for the definition of a tree decomposition.} The width of the tree decomposition constructed in \cite{DBLP:journals/dcg/FominLPSZ19} is $k^{\OO(1)}$, which on its own does not enable to design a subexponential (or even single-exponential) time algorithm. However, each of its bags (of size $k^{\OO(1)}$) is equipped with a partition into $\OO(\sqrt{k})$ sets such that each of them induces a clique. By establishing a property that asserts the existence of a solution (if at least one solution exists) that crosses these cliques ``few'' times, the tree decomposition can be exploited. Specifically, this exploitation requires to design a very tedious dynamic programming algorithm (significantly more technical than algorithms over ``standard'' tree decompositions, that is, tree decompositions of small width) to keep track of the interactions between the cliques in the~partitions. 

We completely avoid the use of the new type of tree decomposition of \cite{DBLP:journals/dcg/FominLPSZ19}. Instead, we use a simple marking procedure to reduce the problem to (a weighted version of) itself on a tree decomposition of width $\OO(\sqrt{k})$. Then, the new problem can be solved by known algorithms as black boxes by employing  either an essentially trivial $\tw^{\OO(\tw)}n$-time algorithm, or a more sophisticated 
$2^{\OO(\tw)}n$-time algorithm (of \cite{DBLP:journals/iandc/BodlaenderCKN15} or \cite{DBLP:journals/jacm/FominLPS16}).  On a high level, we are able to mark few vertices in certain cliques (which become the cliques in the above mentioned partitions of bags in \cite{DBLP:journals/dcg/FominLPSZ19}), so that there exists a solution (if at least one solution exists) that uses only marked vertices as ``portals''---namely, it crosses cliques only via edges whose both endpoints are marked. Then, in each clique, we can just replace all unmarked vertices by a single weighted vertex. This reduces the size of each clique to be constant, and yields a tree decomposition of width $\OO(\sqrt{k})$. We believe that our idea of identification of portals and replacement of all non-portals by few weighted vertices will find further applications in the design of ETH-tight parameterized algorithms on intersection graphs of fat geometric~objects.

Before we turn to briefly survey some additional related works, we would like to stress that shaving off logarithmic factors in the exponent of running times of parameterized algorithms is a major issue in this field. Indeed, when they appear in the exponent, logarithmic factors have a {\em critical} effect on efficiency that can render algorithms impractical even on small instances. Over the past two decades, most fundamental techniques in Parameterized Complexity targeted not only the objective of eliminating the logarithmic factors, but even improving the base $c$ in running times of the form $c^kn^{\OO(1)}$. For example, this includes the aforementioned color coding technique~\cite{AlonYZ} that was developed to shave off the $
\log k$ in a previous $2^{\OO(k\log k)}n^{\OO(1)}$-time algorithm, which further entailed a flurry of research on techniques to improve the base of the exponent (see Section \ref{sec:relatedWorks}), and the cut-and-count technique to design parameterized algorithms in time $2^{\OO(t)}n^{\OO(1)}$ rather than $2^{\OO(t\log t)}n^{\OO(1)}$ (in fact, for connectivity problems such as {\sc Long Path}) on graphs of treewidth $t$ \cite{DBLP:conf/focs/CyganNPPRW11}. Accompanying this active line of research, much efforts were devoted to prove that problems that have long resisted the design of algorithms without logarithmic factors in the exponent are actually unlikely to admit such algorithms~\cite{DBLP:journals/siamcomp/LokshtanovMS18}.

\subsection{Related Works on Long Path and Long Cycle}\label{sec:relatedWorks}

We now briefly survey some known results in Parameterized Complexity on {\sc Long Path} and {\sc Long Cycle}. Clearly, this survey is illustrative rather than comprehensive.  The standard parameterization of {\sc Long Path} and {\sc Long Cycle} is by the solution size $k$, and here we will survey only results that concern this parameterization.

The \probKPath\ (parameterized by the solution size $k$ on general graphs) is arguably one of the five (or even fewer) problems with the richest history in Parameterized Complexity, having parameterized algorithms continuously developed since the early days of the field and until this day. The algorithms developed along the way gave rise to some of the most central techniques in the field, such as color-coding~\cite{AlonYZ} and its incarnation as divide-and-color~\cite{chen2009randomized}, techniques based on the polynomial method   \cite{Koutis08,KoutisW16,Williams09,DBLP:journals/jcss/BjorklundHKK17},  and matroid based techniques \cite{FominLPS16}. The first parameterized algorithm for this problem was an $2^{\OO(k\log k)}n^{\OO(1)}$-time given in 1985 by Monien~\cite{monien1985find}, even before the term ``parameterized algorithm'' was in known use. Originally in 1994, the logarithmic factor was shaved off~\cite{AlonYZ}, resulting in an algorithm with running time $c^{k}n^{\OO(1)}$ for $c=2e$. After that, a long line of works that presented improvements over $c$ has followed~\cite{Koutis08,KoutisW16,Williams09,DBLP:journals/jcss/BjorklundHKK17,FominLPS16,Zehavi14,DBLP:journals/algorithmica/HuffnerWZ08,DBLP:journals/jcss/ShachnaiZ16,chen2009randomized,DBLP:journals/tcs/Tsur19b}, where the algorithm with the current best running time is a randomized algorithm whose time complexity is $1.66^k n^{\cO(1)}$~\cite{DBLP:journals/jcss/BjorklundHKK17}.
Unless the ETH fails, \probKPath\ (as well as {\sc Long Cycle}) does not admit any algorithm with running time $2^{o(k)} n^{\cO(1)}$~\cite{ImpagliazzoPZ01}.

For a long time, the {\sc Long Cycle} problem was considered to be significantly harder than {\sc Long Path} due to the following reason: while the existence of a path of size at least $k$ implies the existence of a path of size exactly $k$, the existence of a cycle of size at least $k$ does not imply the existence of a cycle of size exactly $k$---in fact, the only cycle of size at least $k$ in the input graph might be a Hamiltonian cycle. Thus, for this problem, the first parameterized algorithm appeared (originally) only in 2004~\cite{gabow2008finding}, and the first parameterized algorithm with running time $2^{\OO(k)}n^{\OO(1)}$ appeared (originally) only in 2014~\cite{FominLPS16}. Further improvements on the base of the exponent in the running time were given in \cite{DBLP:journals/ipl/Zehavi16,DBLP:journals/ipl/FominLPSZ18}. Lastly, we remark that due to their importance, over the past two decades there has been extensive research of {\sc Long Path} and {\sc Long Cycle} parameterized by $k$ above some guarantee~\cite{DBLP:conf/icalp/BezakovaCDF17,DBLP:conf/esa/FominGLP0Z19,DBLP:conf/wg/Jansen0N19,ITCS20Girth}, and the (approximate) counting versions of these problems~\cite{FlumG04,DBLP:conf/isaac/ArvindR02,DBLP:conf/ismb/AlonDHHS08,DBLP:journals/talg/AlonG10,DBLP:conf/iwpec/AlonG09,DBLP:conf/stoc/BrandDH18,DBLP:conf/icalp/BjorklundL0Z19}. Both {\sc Long Path} and {\sc Long Cycle} are unlikely to admit a polynomial kernel~\cite{bodlaender2009problems}, and in fact, are even conjectured not to admit Turing kernels~\cite{DBLP:journals/algorithmica/HermelinKSWW15,DBLP:journals/algorithmica/JansenPW19}.

While \probKPath and \probKCycle  remain NP-complete on planar graphs, they admit  $2^{\cO(\sqrt{k})} n^{\cO(1)}$-time algorithms: By combining  the bidimensionality theory  of Demaine et al. \cite{DemaineFHT05jacm} with efficient   algorithms on graphs of bounded treewidth \cite{DornPBF10}, \probKPath\ and \probKCycle, can be solved in time $2^{\cO(\sqrt{k})} n^{\cO(1)}$ on planar graphs. Moreover, the parameterized subexponential ``tractability'' of \probKPath/{\sc Cycle} can be extended to  graphs excluding some fixed graph as a minor~\cite{Demaine:2008mi}. Unfortunately, unit disk graphs are somewhat different than planar graphs and $H$-minor free graphs---in particular, unlike planar graphs and $H$-minor free graphs where the maximum clique size is bounded by $5$ (for planar graphs) or some other fixed constant (for $H$-minor free graphs), unit disk graphs can contain cliques of arbitrarily large size and are therefore ``highly non-planar''. Nevertheless, Fomin et al.~\cite{FominLS12}  were able to obtain subexponential parameterized algorithms of running time $2^{\cO(k^{0.75}\log{k})} n^{\cO(1)}$ on unit disk graphs for \probKPath, \probKCycle, \probFVS and  \probCycPacking. None of these four problems can be solved in time  $2^{o(\sqrt{n})}$(and hence also in time $2^{o(\sqrt{k})}n^{\OO(1)}$) on unit disk graphs unless the \ETH\ fails~\cite{DBLP:conf/stoc/BergBKMZ18}.
Afterwards (originally in 2017), Fomin et al.~\cite{DBLP:journals/dcg/FominLPSZ19} obtained improved, yet technically quite tedious, $2^{\cO(\sqrt{k}\log{k})} n^{\cO(1)}$-time algorithms for  
\probKPath, \probKCycle and \probFVS\ and \probCycPacking. Recall that this work was discussed earlier in the introduction. 
Later, the same set of authors designed  $2^{\cO(\sqrt{k}\log{k})}  n^{\cO(1)}$ time algorithm  for the aforementioned problems on map graphs~\cite{DBLP:conf/icalp/FominLP0Z19}. We also remark that recently, Panolan et al.~\cite{DBLP:conf/soda/Panolan0Z19} proved a contraction decomposition theorem on unit disk graphs and as an application of the theorem, they proved that {\sc Min-Bisection} on unit disk graphs can be solved in time $2^{\cO(k)}n^{\cO(1)}$.

\section{Preliminaries}\label{sec:prelims}

For $\ell\in\mathbb{N}$, let $[\ell]=\{1,\ldots,\ell\}$. 
For a graph $G$, let $V(G)$ and $E(G)$ denote its vertex set and edge set, respectively. When $G$ is clear from context, let $n=|V(G)|$ and $m=|E(G)|$. For a subset $U\subseteq V(G)$, let $G[U]$ denote the subgraph of $G$ induced by $U$. A graph $H$ is a {\em minor} of $G$ if $H$ can be obtained from $G$ by a sequence of edge deletions, edge contractions and vertex deletions. Given $a,b\in\mathbb{N}$, an {\em $a\times b$-grid} is a graph on $a\cdot b$ vertices that can be denoted by $v_{i,j}$ for $(i,j)\in[a]\times[b]$, such that $E(G)=\{\{v_{i,j},v_{i+1,j}\}: i\in[a-1],j\in[b]\}\cup\{\{v_{i,j},v_{i,j+1}\}: i\in[a], j\in[b-1]\}$.

\subparagraph*{Unit disk graphs.}
Let $P=\{p_1=(x_1,y_1),p_2=(x_2,y_2),\ldots,p_n=(x_n,y_n)\}$ be a set of points in the Euclidean plane. Let $D=\{d_1,d_2,\ldots,d_n\}$ where for every $i\in [n]$, $d_i$ is the disk of radius 1 whose centre is $p_i$.  Then, the unit disk graph of $D$ is the graph $G$ such that $V(G)=D$ and $E(G)=\{\{d_i,d_j\}~|~d_i,d_j\in D, i\neq j, \sqrt{(x_i-x_j)^2+(y_i-y_j)^2}\leq 2\}$. Throughout the paper, we suppose that any given UDG $G$ is accompanied by a corresponding set of disks $D$, which is critical to our algorithm. We also remark that, given  a graph $G$ (without $D$), the decision of whether $G$ is a UDG is NP-hard~\cite{breu1998unit} (in fact, even $\exists\mathbb{R}$-hard~\cite{DBLP:journals/dcg/KangM12}). 

\subparagraph*{Clique-Grids.} Intuitively, a clique-grid is a graph whose vertices can be embedded in grid cells (where multiple vertices can be embedded in each cell), so that the each cell induces a clique and ``interacts'' (via edges incident to its vertices) only with cells at ``distance'' at most~2 (see Fig.~\ref{fig:cliquegridgraph}).

\begin{definition}[{\bf Clique-Grid}]\label{def:GridClique}
A graph $G$ is a {\em clique-grid} if there exists a function $f: V(G)\rightarrow [t]\times [t]$  for some $t\in {\mathbb N}$, called a {\em representation}, such that the following conditions are satisfied.
\begin{enumerate}\setlength\itemsep{0em} 
\item\label{condition:GridClique1} For all $(i,j)\in [t]\times [t]$, $G[f^{-1}(i,j)]$ is a clique.
\item\label{condition:GridClique2} For all $\{u,v\}\in E(G)$, $|i-i'|\leq 2$ and $|j-j'|\leq 2$ where $f(u)=(i,j)$ and $f(v)=(i',j')$.
\end{enumerate}
\end{definition}
We call a pair $(i,j)\in [t]\times [t]$ a {\em cell}. It is easy to see that a unit disk graph is a clique-grid, and a representation of it, can be computed in linear time. A formal proof can be found in \cite{DBLP:journals/dcg/FominLPSZ19} (also see \cite{ito2010tractability} for a similar result). Specifically, we will refer to the following proposition.

\begin{figure}
\begin{center}
\begin{tikzpicture}[scale=2]
\draw (0,1)--(6,1);
\draw (0,2)--(6,2);
\draw (0,3)--(6,3);
\draw (0,4)--(6,4);
\draw (0,5)--(6,5);
\draw (0,1)--(0,5);
\draw (1,1)--(1,5);
\draw (2,1)--(2,5);
\draw (3,1)--(3,5);
\draw (4,1)--(4,5);
\draw (5,1)--(5,5);
\draw (6,1)--(6,5);
\draw (2.8,2.8)--(2.7,2.3)--(2.2,2.85)--(2.1,2.1)--(2.8,2.8)--(2.2,2.85);
\draw (2.1,2.1)--(2.7,2.3);
\node[] (a1) at (2.8,2.8) {$\bullet$}; 
\node[] (a1) at (2.7,2.3) {$\bullet$}; 
\node[] (a1) at (2.2,2.85) {$\bullet$}; 
\node[] (a1) at (2.1,2.1) {$\bullet$}; 
\draw[blue] (1.7,2.3)--(2.2,2.85)--(1.2,2.85);
\draw[blue]  (1.7,2.3)--(2.1,2.1)to [out=90,in=-20] (1.2,2.85);
\draw[blue]  (1.2,2.1)--(2.1,2.1)--(0.9,1.8);
\draw (1.2,2.1)--(1.7,2.3)--(1.2,2.85)--(1.2,2.1);
\node[] (a1) at (1.2,2.1) {$\bullet$}; 
\node[] (a1) at (1.7,2.3) {$\bullet$}; 
\node[] (a1) at (1.2,2.85) {$\bullet$}; 
\draw[blue] (1.2,2.85)--(0.9,1.8);
\draw[blue]  (0.1,1.85)--(1.2,2.1)--(0.9,1.8);
\draw[blue]  (1.2,2.1)--(0.8,1.1);
\draw (0.8,1.1) --(0.9,1.8) -- (0.1,1.85) -- (0.8,1.1);
\node[] (a1) at (0.8,1.1) {$\bullet$}; 
\node[] (a1) at (0.9,1.8) {$\bullet$}; 
\node[] (a1) at (0.1,1.85) {$\bullet$}; 
\draw (3.8,3.1)--(3.7,3.7);
\node[] (a1) at (3.8,3.1) {$\bullet$}; 
\node[] (a1) at (3.7,3.7) {$\bullet$}; 
\draw[blue] (2.8,2.8)--(3.8,3.1)--(2.7,2.3);
\draw[blue] (2.8,2.8)--(3.7,3.7)--(2.7,2.3);
\draw (4.7,2.1) --(4.4,2.8)--(4.1,2.85)--(4.7,2.1);
\node[] (a1) at (4.7,2.1) {$\bullet$}; 
\node[] (a1) at (4.4,2.8) {$\bullet$}; 
\node[] (a1) at (4.1,2.85) {$\bullet$}; 
\draw[blue] (3.8,3.1)--(4.1,2.85) -- (3.7,3.7);
\draw[blue] (3.8,3.1)--(4.4,2.8)-- (3.7,3.7);
\draw (5.7,4.9) -- (5.1,4.1)--(5.9,4.7)--(5.2,4.5)--(5.1,4.1);
\draw (5.2,4.5) -- (5.7,4.9)--(5.9,4.7);
\node[red] (a1) at (5.9,4.7) {$\bullet$};
\node[red] (a1) at (5.2,4.5) {$\bullet$};
\node[red] (a1) at (5.7,4.9) {$\bullet$}; 
\node[] (a1) at (5.1,4.1) {$\bullet$}; 
\draw[blue] (5.1,4.1)--(4.4,2.8);
\draw[blue] (5.1,4.1)--(3.7,3.7);
\draw[blue] (4.7,2.1)-- (5.1,1.5);
\draw[blue] (4.7,2.1)-- (5.3,1.7);
\draw[blue] (4.7,2.1)-- (5.6,1.7);
\draw[blue] (4.7,2.1)-- (5.8,1.5);
\draw[blue] (4.7,2.1)-- (5.6,1.2);
\draw[red] (4.7,2.1)-- (5.2,1.1);
\draw (5.1,1.5) --(5.3,1.7)--(5.6,1.7)--(5.8,1.5)--(5.6,1.2)--(5.2,1.1)-- (5.1,1.5)--(5.6,1.7)--(5.6,1.2)-- (5.1,1.5)--(5.8,1.5)--(5.2,1.1)--(5.3,1.7)--(5.8,1.5);
\draw (5.3,1.7)--(5.6,1.2);
\draw (5.2,1.1)--(5.6,1.7);
\draw (5.4,1.05) -- (5.1,1.5) ;
\draw (5.4,1.05) -- (5.3,1.7) ;
\draw (5.4,1.05) -- (5.6,1.7)  ;
\draw (5.4,1.05) --  (5.8,1.5) ;
\draw (5.4,1.05) --  (5.6,1.2) ;
\draw (5.4,1.05) --  (5.2,1.1) ;
\draw[blue](5.9,2.8)--(5.6,1.7);
\node[] (a1) at (5.1,1.5) {$\bullet$}; 
\node[] (a1) at (5.3,1.7) {$\bullet$}; 
\node[] (a1) at (5.6,1.7) {$\bullet$}; 
\node[] (a1) at (5.8,1.5) {$\bullet$}; 
\node[] (a1) at (5.6,1.2) {$\bullet$}; 
\node[red] (a1) at (5.2,1.1) {$\bullet$}; 
\node[red] (a1) at (5.4,1.05) {$\bullet$}; 
\node[] (a1) at (5.9,2.8) {$\bullet$}; 
\end{tikzpicture}
\end{center}
\caption{A clique-grid graph $G$. Marked vertices are colored black. Good and bad edges are colored blue and red, respectively (see Definition~\ref{def:goodEdge}).
For the sake of illustration only, we use the threshold 5 instead of 121---that is, in phase II of marking let $\ma_2(v,(i',j'))$ denote a set of $5$ vertices in $f^{-1}(i',j')$ that are adjacent to $v$ in $G$, where if no $5$ vertices with this property exist, then let $\ma_2(v,(i',j'))$ denote the set of all vertices with this property.}
\label{fig:cliquegridgraph}
\end{figure}
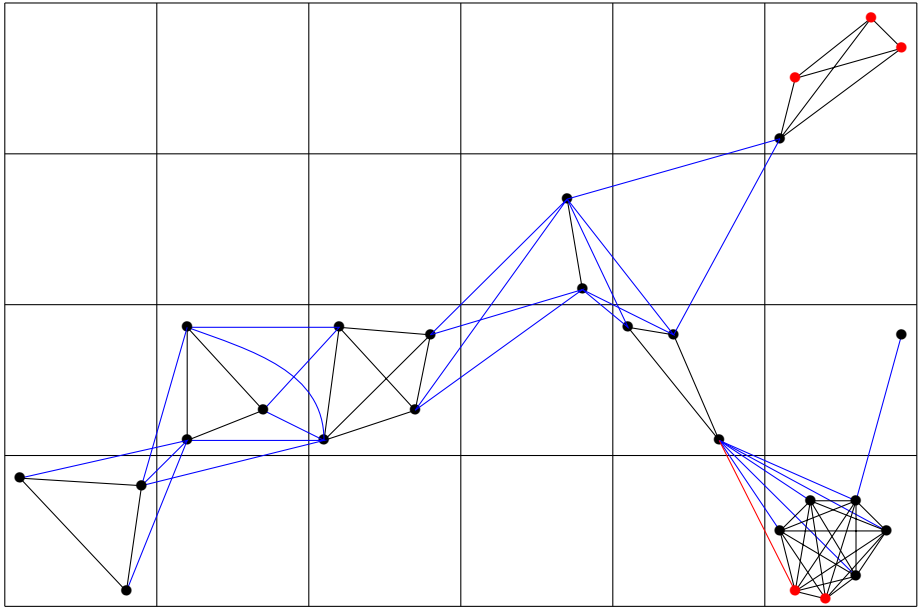

\begin{proposition}[\cite{ito2010tractability,DBLP:journals/dcg/FominLPSZ19}]\label{prop:cliqueGrid}
Let $G$ be the unit disk graph of a set of unit disks $D$ in the Euclidean plane. Then, $G$ is a clique-grid, and a representation of $G$ can be computed in linear time.
\end{proposition}

\subparagraph*{Treewidth.}  The treewidth of a graph is a standard measure of its ``closeness''  to a tree, formally defined as follows. 

\begin{definition}[{\bf Treewidth}]\label{def:treeDecomp}
A {\em tree decomposition} of a graph $G$ is a pair $(T,\beta)$, where $T$ is a tree and $\beta$ is a function from $V(T)$ to $2^{V(G)}$, that satisfies the following conditions.
\begin{itemize}\setlength\itemsep{0em} 
\item For every edge $\{u,v\}\in E(G)$, there exists $x\in V(T)$ such that  $\{u,v\}\subseteq \beta(x)$.
\item For every vertex $v\in V(G)$, $T[\{x\in V(T)\}]$ is a tree on at least one vertex.  
\end{itemize}
The {\em width} of $(T,\beta)$ is $\max_{x\in V(T)} |\beta(x)|-1$. The {\em treewidth} of $G$, denoted by $\tw(G)$, is the minimum width over all tree decompositions of $G$.
\end{definition}

The treewidth of a graph can be approximated within a constant factor efficiently as follows.

\begin{proposition}[\cite{BodlaenderDDFLP16}]\label{prop:treewidth}
Given a graph $G$ and a positive integer $k$, in time $2^{\OO(k)} \cdot n$, we can either decide that $\tw(G)>k$ or 
output a tree decomposition of $G$ of width $5k$.
\end{proposition}

We will need the following proposition to argue that a unit disk graph {\em of bounded degree} contains a grid minor of dimension {\em linear} in its treewidth.

\begin{proposition}[\cite{FominLS12}]\label{prop:gridUnitDisk}
Let $G$ be a unit disk graph with maximum degree $\Delta$ and treewidth $\tw$. Then, $G$ contains a $\displaystyle{\frac{\tw}{100\Delta^3}\times\frac{\tw}{100\Delta^3}}$ grid as a minor.
\end{proposition}

\section{Marking Scheme}\label{sec:marking}

In this section, we present a marking scheme whose purpose is to mark a constant number of vertices in each cell of a clique-grid $G$ so that, if $G$ has a path (resp.~cycle) on at least $k$ vertices, then it also has a path (resp.~cycle) on at least $k$ vertices that ``crosses'' cells only at marked vertices. Then, we further argue that unmarked vertices in a cell can be thought of, in a sense, as a ``unit'' representable by one weighted vertex. We note that we did not make any attempt to optimize the number of vertices marked, but only make the proof simple.

\subparagraph*{Marking Scheme.}
Let $G$ be a clique-grid graph with representation $f: V(G)\rightarrow [t]\times [t]$. Then, the marking scheme consists of two phases defined as follows.

\smallskip
\noindent{\bf Phase I.}  For each pair of distinct cells $(i,j),(i',j')\in [t]\times [t]$ with $|i-i'|\leq 2$ and $|j-j'|\leq 2$, let $M$ be a maximal matching where each edge has one endpoint in $f^{-1}(i,j)$ and the other endpoint in $f^{-1}(i',j')$; if $|M|\leq 241$, then denote $\ma_1(\{(i,j),(i',j')\})=M$, and otherwise choose a subset $M'$ of $M$ of size $241$ and let $\ma_1(\{(i,j),(i',j')\})=M'$.

For each cell $(i,j)\in [t]\times[t]$, let $\ma_1(i,j)$ denote the set of all vertices in $f^{-1}(i,j)$ that are endpoints of edges in $\bigcup_{(i',j')}\ma_1(\{(i,j),(i',j')\})$ where $(i',j')$ ranges over all cells such that $|i-i'|\leq 2$ and $|j-j'|\leq 2$; the vertices belonging to this set are called {\em marked vertices}.

\smallskip
\noindent{\bf Phase II.} For each ordered pair of distinct cells $(i,j),(i',j')\in[t]\times[t]$ with $|i-i'|\leq 2$ and $|j-j'|\leq 2$ and vertex $v\in\ma_1(i,j)$, let $\ma_2(v,(i',j'))$ denote a set of $121$ 
vertices in $f^{-1}(i',j')$ that are adjacent to $v$ in $G$, where if no $121$ vertices with this property exist, then let $\ma_2(v,(i',j'))$ denote the set of all vertices with this property; the vertices that belong to this set are also called {\em marked vertices}.

\smallskip
\noindent{\bf Altogether.} For each cell $(i,j)\in[t]\times[t]$, let $\ma^\star(i,j)$ denote the set of all marked vertices in $f^{-1}(i,j)$.

\bigskip
\noindent Clearly, given $G$ and $f$, $\ma^\star(i,j)$ is not uniquely defined. Whenever we write $\ma^\star(i,j)$, we refer to an arbitrary set that can be the result of the scheme above. We remark that ``guessing'' the endpoints of the sought path and marking them can simplify some arguments ahead, but this will lead to worse time complexity. We have the following simple observation regarding the size of $\ma^\star(i,j)$ and the computation time.

\begin{observation}\label{obs:sizeMa}
Let $G$ be a clique-grid with representation $f: V(G)\rightarrow[t]\times[t]$. For each cell $(i,j)\in[t]\times[t]$,$|\ma^\star(i,j)|\leq 10^{10}$. Moreover, the computation of all the sets $\ma^\star(i,j)$ together can be done in linear time.
\end{observation}

\begin{proof}
Consider a cell $(i,j)\in[t]\times[t]$. In the first phase, at most $24\cdot 241$ vertices in $f^{-1}(i,j)$ are marked. In the second phase, for each of the $24$ cells $(i',j')$ such that $|i-i'|\leq 2$ and $|j-j'|\leq 2$, and each of the at most $24\cdot 241$ marked vertices in $f^{-1}(i',j')$, at most $121$ new vertices in $f^{-1}(i,j)$ are marked. Therefore, in total at most $24\cdot 241+24\cdot (24\cdot 241)\cdot 121\leq 10^{10}$ vertices in $f^{-1}(i,j)$ are marked.
The claim regarding the computation time is immediate.
\end{proof}

As part of the proof that our marking scheme has the property informally stated earlier, we will use the following proposition.

\begin{proposition}[\cite{DBLP:journals/dcg/FominLPSZ19}]\label{prop:fewCrossings}
Let $G$ be a clique-grid with representation $f$ that has a path (resp.~cycle) on at least $k$ vertices. Then, $G$ also has a path (resp.~cycle) $P$ on at least $k$ vertices with the following property: for every two distinct cells $(i,j)$ and $(i',j')$, there exist at most $5$ edges $\{u,v\}\in E(P)$ such that $f(u)=(i,j)$ and $f(v)=(i',j')$.
\end{proposition}

We now formally state and prove the property achieved by our marking scheme. For this purpose, we have the following definition (see Fig.~\ref{fig:cliquegridgraph}) and lemma.

\begin{definition}\label{def:goodEdge}
Let $G$ be a clique-grid with representation $f$. 
An edge $\{u,v\}\in E(G)$ where $f(u)\neq f(v)$ is a {\em good} edge if $u\in\ma^\star(i,j)$ and $v\in\ma^\star(i',j')$ where $f(u)=(i,j)$ and $f(v)=(i',j')$; otherwise, it is {\em bad}.
\end{definition}

Intuitively, the following lemma asserts the existence of a solution (if any solution exists) that crosses different cells only via good edges, that is, via marked vertices.

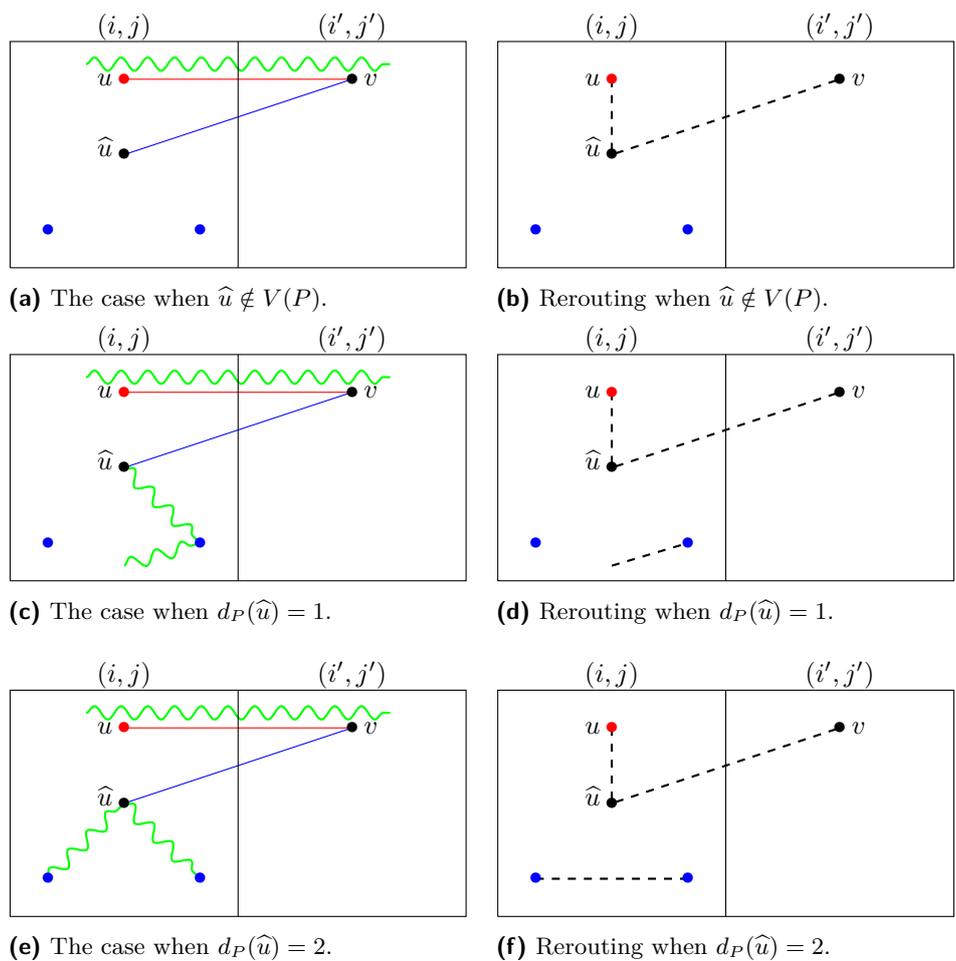
\begin{figure}
\begin{center}
\begin{subfigure}[b]{0.45\textwidth}
\begin{tikzpicture}[scale=1]
 \path [draw=green,snake it,thick] (1,2.7)--(5,2.7);
\draw[red] (1.5,2.5)--(4.5,2.5);
\draw[blue] (1.5,1.5)--(4.5,2.5);
\node[] (a1) at (1.5,3.2) {$(i,j)$}; 
\node[] (a1) at (4.5,3.2) {$(i',j')$};
\node[] (a1) at (1.25,2.5) {$u$};
\node[red] (a1) at (1.5,2.5) {$\bullet$}; 
\node[] (a1) at (1.5,1.5) {$\bullet$}; 
\node[] (a1) at (1.25,1.6) {$\widehat{u}$}; 
\node[blue] (a1) at (0.5,0.5) {$\bullet$}; 
\node[blue] (a1) at (2.5,0.5) {$\bullet$}; 
\node[] (a1) at (4.75,2.5) {$v$};
\node[] (a1) at (4.5,2.5) {$\bullet$};
\draw (0,0)--(0,3)--(6,3)--(6,0)--(0,0);
\draw (3,0)--(3,3);
\end{tikzpicture}
\caption{The case when $\widehat{u}\notin V(P)$.}
\end{subfigure}
\begin{subfigure}[b]{0.45\textwidth}
\begin{tikzpicture}[scale=1]
\draw[dashed, thick] (1.5,2.5)--(1.5,1.5)--(4.5,2.5);
\node[] (a1) at (1.5,3.2) {$(i,j)$}; 
\node[] (a1) at (4.5,3.2) {$(i',j')$};
\node[] (a1) at (1.25,2.5) {$u$};
\node[red] (a1) at (1.5,2.5) {$\bullet$}; 
\node[] (a1) at (1.5,1.5) {$\bullet$}; 
\node[blue] (a1) at (0.5,0.5) {$\bullet$}; 
\node[blue] (a1) at (2.5,0.5) {$\bullet$}; 
\node[] (a1) at (1.25,1.6) {$\widehat{u}$}; 
\node[] (a1) at (4.75,2.5) {$v$};
\node[] (a1) at (4.5,2.5) {$\bullet$};
\draw (0,0)--(0,3)--(6,3)--(6,0)--(0,0);
\draw (3,0)--(3,3);
\end{tikzpicture}
\caption{Rerouting when $\widehat{u}\notin V(P)$.}
\end{subfigure}
\vspace{0.3cm}
\begin{subfigure}[b]{0.45\textwidth}
\begin{tikzpicture}[scale=1]
 \path [draw=green,snake it,thick] (1,2.7)--(5,2.7); 
 \path [draw=green,snake it,thick] (1.5,1.5)--(2.5,0.5);
 \path [draw=green,snake it,thick](1.5,0.2)--(2.5,0.5);
\draw[red] (1.5,2.5)--(4.5,2.5);
\draw[blue] (1.5,1.5)--(4.5,2.5);
\node[] (a1) at (1.5,3.2) {$(i,j)$}; 
\node[] (a1) at (4.5,3.2) {$(i',j')$};
\node[] (a1) at (1.25,2.5) {$u$};
\node[red] (a1) at (1.5,2.5) {$\bullet$}; 
\node[] (a1) at (1.5,1.5) {$\bullet$}; 
\node[blue] (a1) at (0.5,0.5) {$\bullet$}; 
\node[blue] (a1) at (2.5,0.5) {$\bullet$}; 
\node[] (a1) at (1.25,1.6) {$\widehat{u}$}; 
\node[] (a1) at (4.75,2.5) {$v$};
\node[] (a1) at (4.5,2.5) {$\bullet$};
\draw (0,0)--(0,3)--(6,3)--(6,0)--(0,0);
\draw (3,0)--(3,3);
\end{tikzpicture}
\caption{The case when $d_{P}(\widehat{u})=1$.}
\end{subfigure}
\begin{subfigure}[b]{0.45\textwidth}
\begin{tikzpicture}[scale=1]
\draw[dashed, thick] (1.5,2.5)--(1.5,1.5)--(4.5,2.5);
\draw[dashed, thick](1.5,0.2)--(2.5,0.5);
\node[] (a1) at (1.5,3.2) {$(i,j)$}; 
\node[] (a1) at (4.5,3.2) {$(i',j')$};
\node[] (a1) at (1.25,2.5) {$u$};
\node[red] (a1) at (1.5,2.5) {$\bullet$}; 
\node[] (a1) at (1.5,1.5) {$\bullet$}; 
\node[blue] (a1) at (0.5,0.5) {$\bullet$}; 
\node[blue] (a1) at (2.5,0.5) {$\bullet$}; 
\node[] (a1) at (1.25,1.6) {$\widehat{u}$}; 
\node[] (a1) at (4.75,2.5) {$v$};
\node[] (a1) at (4.5,2.5) {$\bullet$};
\draw (0,0)--(0,3)--(6,3)--(6,0)--(0,0);
\draw (3,0)--(3,3);
\end{tikzpicture}
\caption{Rerouting when $d_{P}(\widehat{u})=1$.}
\end{subfigure}
\vspace{0.3cm}
\begin{subfigure}[b]{0.45\textwidth}
\begin{tikzpicture}[scale=1]
 \path [draw=green,snake it,thick] (1,2.7)--(5,2.7);
 \path [draw=green,snake it,thick](0.5,0.5)--(1.5,1.5);
 \path [draw=green,snake it,thick](1.5,1.5)--(2.5,0.5);
\draw[red] (1.5,2.5)--(4.5,2.5);
\draw[blue] (1.5,1.5)--(4.5,2.5);
\node[] (a1) at (1.5,3.2) {$(i,j)$}; 
\node[] (a1) at (4.5,3.2) {$(i',j')$};
\node[] (a1) at (1.25,2.5) {$u$};
\node[red] (a1) at (1.5,2.5) {$\bullet$}; 
\node[] (a1) at (1.5,1.5) {$\bullet$}; 
\node[blue] (a1) at (0.5,0.5) {$\bullet$}; 
\node[blue] (a1) at (2.5,0.5) {$\bullet$}; 
\node[] (a1) at (1.25,1.6) {$\widehat{u}$}; 
\node[] (a1) at (4.75,2.5) {$v$};
\node[] (a1) at (4.5,2.5) {$\bullet$};
\draw (0,0)--(0,3)--(6,3)--(6,0)--(0,0);
\draw (3,0)--(3,3);
\end{tikzpicture}
\caption{The case when $d_{P}(\widehat{u})=2$.}
\end{subfigure}
\begin{subfigure}[b]{0.45\textwidth}
\begin{tikzpicture}[scale=1]
\draw[dashed, thick] (1.5,2.5)--(1.5,1.5)--(4.5,2.5);
\draw[dashed, thick] (0.5,0.5)--(2.5,0.5);
\node[] (a1) at (1.5,3.2) {$(i,j)$}; 
\node[] (a1) at (4.5,3.2) {$(i',j')$};
\node[] (a1) at (1.25,2.5) {$u$};
\node[red] (a1) at (1.5,2.5) {$\bullet$}; 
\node[] (a1) at (1.5,1.5) {$\bullet$}; 
\node[blue] (a1) at (0.5,0.5) {$\bullet$}; 
\node[blue] (a1) at (2.5,0.5) {$\bullet$}; 
\node[] (a1) at (1.25,1.6) {$\widehat{u}$}; 
\node[] (a1) at (4.75,2.5) {$v$};
\node[] (a1) at (4.5,2.5) {$\bullet$};
\draw (0,0)--(0,3)--(6,3)--(6,0)--(0,0);
\draw (3,0)--(3,3);
\end{tikzpicture}
\caption{Rerouting when $d_{P}(\widehat{u})=2$.}
\end{subfigure}
\end{center}
\caption{Case I in the proof of Lemma~\ref{lem:propertyMaMain}. Vertices colored black and red are marked and unmarked, respectively. Vertices colored blue are either marked or unmarked. Good and bad edges are colored blue and red, respectively. Curves colored green are part of the path $P$. Dashed lines are part of the path $P_2$.}
\label{fig:caseone}
\end{figure}

\begin{lemma}\label{lem:propertyMaMain}
Let $G$ be a clique-grid with representation $f$ that has a path (resp.~cycle) on at least $k$ vertices. Then, $G$ also has a path (resp.~cycle) $P$ on at least $k$ vertices with the following property: every edge $\{u,v\}\in E(P)$ where $f(u)\neq f(v)$ is a good edge.
\end{lemma}

\begin{proof}
By Proposition \ref{prop:fewCrossings}, $G$ has a path (resp.~cycle) on at least $k$ vertices with the following property: for every two distinct cells $(i,j)$ and $(i',j')$, there exist at most $5$ edges $\{u,v\}\in E(P)$ such that $f(u)=(i,j)$ and $f(v)=(i',j')$.
Among all such paths (resp.~cycles), let $P$ be one that minimizes the number of bad edges. The following claim follows immediately from the choice of $P$ and Property \ref{condition:GridClique2} in Definition \ref{def:GridClique}.

\begin{claim}\label{claim1}
For each cell $(i,j)\in[t]\times[t]$, there are at most $24\cdot 5=120$ vertices in $f^{-1}(i,j)\cap V(P)$ that are adjacent in $P$ to at least one vertex that does not belong to $f^{-1}(i,j)$.
\end{claim}

Next, we show that $P$ has no bad edge, which will complete the proof.
Targeting a contradiction, suppose that $P$ has some bad edge $\{u,v\}$. By Definition \ref{def:goodEdge}, $u\notin\ma^\star(i,j)$ or $v\notin\ma^\star(i',j')$ (or both) where $f(u)=(i,j)$ and $f(v)=(i',j')$. Without loss of generality, suppose that $u\notin\ma^\star(i,j)$. We consider two cases as follows. 

\medskip
\noindent{\bf Case I.} First, suppose that $v\in\ma_1(i',j')$. Because $u$ is adjacent to $v$ but it is not marked in the second phase, it must hold that $|\ma_2(v,(i,j))|\geq 121$. By Claim \ref{claim1}, this means that there exists a vertex $\widehat{u}\in \ma_2(v,(i,j))$ such the vertices incident to it on $P$---which might be $0$ if $\widehat{u}$ does not belong to $P$, $1$ if it is an endpoint of $P$ or $2$ if it is an internal vertex of $P$---also belong to $f^{-1}(i,j)$ (see Fig.~\ref{fig:caseone}). In case $\widehat{u}\notin V(P)$, denote $P_1=P$. Else, by Property \ref{condition:GridClique1} in Definition \ref{def:GridClique}, by removing $\widehat{u}$ from $P$, and if $\widehat{u}$ has two neighbors on $P$, then also making these two neighbors adjacent,\footnote{If $\widehat{u}$ is an endpoint of $P$, then only the removal of $\widehat{u}$ is performed.} we still have a path (resp.~cycle) in $G$; we denote this path by $P_1$. Note that $|V(P_1)|\geq |V(P)|-1$ and $u\notin V(P_1)$. Now, note that because $\widehat{u}\in \ma_2(v,(i,j))$, we have that $\widehat{u}$ is adjacent to $v$ in $G$ and also $\widehat{u}\in f^{-1}(i,j)$. Because $u\in f^{-1}(i,j)$, Property \ref{condition:GridClique1} in Definition \ref{def:GridClique} implies that $\widehat{u}$ is also adjacent to $u$. Thus, by inserting $\widehat{u}$ between $u$ and $v$ in $P_1$ and making it adjacent to both, we still have a path (resp.~cycle) in $G$, which we denote by $P_2$ (see Fig.~\ref{fig:caseone}). Note that $|V(P_2)|=|V(P_1)|+1\geq |V(P)|\geq k$. Moreover, the only edges that appear only in one among $P_2$ and $P$ are as follows.
\begin{enumerate}
\item If $\widehat{u}$ has two neighbors in $P$, then the edges between $\widehat{u}$ and these two neighbors might belong only to $P$, and the edge between these two neighbors belongs only to $P_2$. As $\widehat{u}$ and its neighbors in $P$ belong to the same cell (by the choice of $\widehat{u}$), none of these edges is bad, and also none of these edges crosses different cells.

\item If $\widehat{u}$ has only one neighbor in $P$, then the edge between $\widehat{u}$ and this neighbor might belong only to $P$.

\item $\{u,v\}\in E(P)\setminus E(P_2)$ is a bad edge that crosses different cells by its initial choice.

\item $\{u,\widehat{u}\}$ might belong only to $P_2$, and it is a neither a bad edge nor an edge that crosses different cells because $u$ and $\widehat{u}$ belong to the same cell.
\item $\{\widehat{u},v\}\in E(P_2)\setminus E(P)$ is a not a bad edge because both $\widehat{u}$ and $v$ are marked (since $v\in\ma_1(i',j')$ and $\widehat{u}\in \ma_2(v,(i,j))$), but it crosses different cells.
\end{enumerate}
Thus, $P_2$ has no bad edge that does not belong to $P$, and $P$ has at least one bad edge that does not belong to $P_2$  (specifically, $\{u,v\}$), and therefore $P_2$ has fewer bad edges than $P$. Moreover, notice that the items above also imply that $P_2$ has at most one edge that crosses different cells and does not belong to $P$ (specifically, $\{\widehat{u},v\}$), and $P$ has at least one edge that crosses the {\em same} cells and does not belong to $P_2$ (specifically, $\{u,v\}$). Therefore, $P_2$ also has the property of $P$ that for every two distinct cells $(\widetilde{i},\widetilde{j})$ and $(\widetilde{i}',\widetilde{j}')$, there exist at most $5$ edges $\{\widetilde{u},\widetilde{v}\}\in E(P_2)$ such that $f(\widetilde{u})=(\widetilde{i},\widetilde{j})$ and $f(\widetilde{v})=(\widetilde{i}',\widetilde{j}')$. Therefore, we have reached a contradiction to the minimality of the number of bad edges in our choice of $P$.

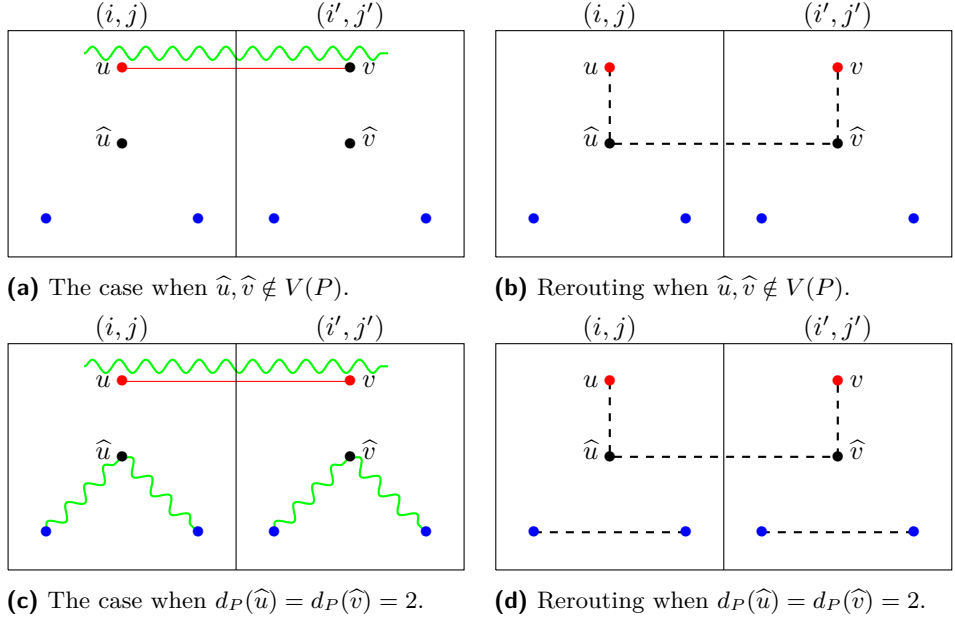
\begin{figure}
\begin{center}
\begin{subfigure}[b]{0.45\textwidth}
\begin{tikzpicture}[scale=1]
 \path [draw=green,snake it,thick] (1,2.7)--(5,2.7);
\draw[red] (1.5,2.5)--(4.5,2.5);
\node[] (a1) at (1.5,3.2) {$(i,j)$}; 
\node[] (a1) at (4.5,3.2) {$(i',j')$};
\node[] (a1) at (1.25,2.5) {$u$};
\node[red] (a1) at (1.5,2.5) {$\bullet$}; 
\node[] (a1) at (1.5,1.5) {$\bullet$}; 
\node[] (a1) at (1.25,1.6) {$\widehat{u}$}; 
\node[blue] (a1) at (0.5,0.5) {$\bullet$}; 
\node[blue] (a1) at (2.5,0.5) {$\bullet$}; 
\node (a1) at (4.5,1.5) {$\bullet$}; 
\node[] (a1) at (4.75,1.6) {$\widehat{v}$}; 
\node[blue] (a1) at (3.5,0.5) {$\bullet$}; 
\node[blue] (a1) at (5.5,0.5) {$\bullet$}; 
\node[] (a1) at (4.75,2.5) {$v$};
\node[] (a1) at (4.5,2.5) {$\bullet$};
\draw (0,0)--(0,3)--(6,3)--(6,0)--(0,0);
\draw (3,0)--(3,3);
\end{tikzpicture}
\caption{The case when $\widehat{u},\widehat{v}\notin V(P)$.}
\end{subfigure}
\begin{subfigure}[b]{0.45\textwidth}
\begin{tikzpicture}[scale=1]
\draw[dashed, thick] (1.5,2.5)--(1.5,1.5)--(4.5,1.5)--(4.5,2.5);
\node[] (a1) at (1.5,3.2) {$(i,j)$}; 
\node[] (a1) at (4.5,3.2) {$(i',j')$};
\node[] (a1) at (1.25,2.5) {$u$};
\node[red] (a1) at (1.5,2.5) {$\bullet$}; 
\node[] (a1) at (1.5,1.5) {$\bullet$}; 
\node[blue] (a1) at (0.5,0.5) {$\bullet$}; 
\node[blue] (a1) at (2.5,0.5) {$\bullet$}; 
\node[] (a1) at (1.25,1.6) {$\widehat{u}$}; 
\node[] (a1) at (4.75,2.5) {$v$};
\node[red] (a1) at (4.5,2.5) {$\bullet$};
\node (a1) at (4.5,1.5) {$\bullet$}; 
\node[] (a1) at (4.75,1.6) {$\widehat{v}$}; 
\node[blue] (a1) at (3.5,0.5) {$\bullet$}; 
\node[blue] (a1) at (5.5,0.5) {$\bullet$}; 
\draw (0,0)--(0,3)--(6,3)--(6,0)--(0,0);
\draw (3,0)--(3,3);
\end{tikzpicture}
\caption{Rerouting when $\widehat{u},\widehat{v}\notin V(P)$.}
\end{subfigure}
\vspace{0.3cm}
\begin{subfigure}[b]{0.45\textwidth}
\begin{tikzpicture}[scale=1]
 \path [draw=green,snake it,thick] (1,2.7)--(5,2.7);
 \path [draw=green,snake it,thick](0.5,0.5)--(1.5,1.5);
 \path [draw=green,snake it,thick](1.5,1.5)--(2.5,0.5);
 \path [draw=green,snake it,thick](3.5,0.5)--(4.5,1.5);
 \path [draw=green,snake it,thick](4.5,1.5)--(5.5,0.5);
\draw[red] (1.5,2.5)--(4.5,2.5);
\node[] (a1) at (1.5,3.2) {$(i,j)$}; 
\node[] (a1) at (4.5,3.2) {$(i',j')$};
\node[] (a1) at (1.25,2.5) {$u$};
\node[red] (a1) at (1.5,2.5) {$\bullet$}; 
\node[] (a1) at (1.5,1.5) {$\bullet$}; 
\node[blue] (a1) at (0.5,0.5) {$\bullet$}; 
\node[blue] (a1) at (2.5,0.5) {$\bullet$}; 
\node[] (a1) at (1.25,1.6) {$\widehat{u}$}; 
\node[] (a1) at (4.75,2.5) {$v$};
\node[red] (a1) at (4.5,2.5) {$\bullet$};
\node (a1) at (4.5,1.5) {$\bullet$}; 
\node[] (a1) at (4.75,1.6) {$\widehat{v}$}; 
\node[blue] (a1) at (3.5,0.5) {$\bullet$}; 
\node[blue] (a1) at (5.5,0.5) {$\bullet$}; 
\draw (0,0)--(0,3)--(6,3)--(6,0)--(0,0);
\draw (3,0)--(3,3);
\end{tikzpicture}
\caption{The case when $d_{P}(\widehat{u})=d_{P}(\widehat{v})=2$.}
\end{subfigure}
\begin{subfigure}[b]{0.45\textwidth}
\begin{tikzpicture}[scale=1]
\draw[dashed, thick] (1.5,2.5)--(1.5,1.5)--(4.5,1.5)--(4.5,2.5);
\draw[dashed, thick] (0.5,0.5)--(2.5,0.5);
\draw[dashed, thick] (3.5,0.5)--(5.5,0.5);
\node[] (a1) at (1.5,3.2) {$(i,j)$}; 
\node[] (a1) at (4.5,3.2) {$(i',j')$};
\node[] (a1) at (1.25,2.5) {$u$};
\node[red] (a1) at (1.5,2.5) {$\bullet$}; 
\node[] (a1) at (1.5,1.5) {$\bullet$}; 
\node[blue] (a1) at (0.5,0.5) {$\bullet$}; 
\node[blue] (a1) at (2.5,0.5) {$\bullet$}; 
\node[] (a1) at (1.25,1.6) {$\widehat{u}$}; 
\node[] (a1) at (4.75,2.5) {$v$};
\node[red] (a1) at (4.5,2.5) {$\bullet$};
\node (a1) at (4.5,1.5) {$\bullet$}; 
\node[] (a1) at (4.75,1.6) {$\widehat{v}$}; 
\node[blue] (a1) at (3.5,0.5) {$\bullet$}; 
\node[blue] (a1) at (5.5,0.5) {$\bullet$}; 
\draw (0,0)--(0,3)--(6,3)--(6,0)--(0,0);
\draw (3,0)--(3,3);
\end{tikzpicture}
\caption{Rerouting when $d_{P}(\widehat{u})=d_{P}(\widehat{v})=2$.}
\end{subfigure}
\end{center}
\caption{Two subcases of Case II in the proof of Lemma~\ref{lem:propertyMaMain}. Other subcases are handled similarly to the subcases depicted here. Vertices colored black and red are marked and unmarked, respectively. Vertices colored blue are either marked or unmarked. Good and bad edges are colored blue and red, respectively. Curves colored green are part of the path $P$. Dashed lines are part of the path $P_2$.}
\label{fig:casetwo}
\end{figure}

\medskip
\noindent{\bf Case II.} Second, suppose that $v\notin\ma^\star(i',j')$. Then, the addition of $\{u,v\}$ to $\ma_1(\{(i,j),$ $(i',j')\})$ maintains the property that it is a matching. Therefore, because this edge was not marked in the first phase, it must hold that  $|\ma_1(\{(i,j),(i',j')\})|=241$. By Claim \ref{claim1}, there are at most $120$ vertices in $f^{-1}(i,j)\cap V(P)$ that are adjacent in $P$ to at least one vertex that does not belong to $f^{-1}(i,j)$, and notice that $u$ (which is unmarked) is one of them. Similarly, there are at most $120$ vertices in $f^{-1}(i',j')\cap V(P)$ that are adjacent in $P$ to at least one vertex that does not belong to $f^{-1}(i',j')$, and notice that $v$ (which is unmarked) is one of them. Therefore, because $\ma_1(\{(i,j),(i',j')\})$ is a matching, it must contain at least one edge $\{\widehat{u},\widehat{v}\}$ such that neither $\widehat{u}$ nor $\widehat{v}$ has a neighbor in $P$ that belongs to a different cell than itself (see Fig.~\ref{fig:casetwo})---either because $\widehat{u}$ (and in the same way $\widehat{v}$) does not belong to $P$, or it does and all its (one or two) neighbors belong to the same cell as itself. Define $P'_1$ as follows: if $\widehat{u}$ does not belong to $P$, then $P'_1=P$, and otherwise let it be the graph obtained by removing $\widehat{u}$  from $P$ and making its two neighbors (if both exist) adjacent. Because these two neighbors (if they exist) belong to the same cell, Property \ref{condition:GridClique1} in Definition \ref{def:GridClique} implies that $P_1'$ is a path (resp.~cycle) in $G$. Similarly, let $P_1$ be the path (resp.~cycle) obtained by the same operation with respect to $P_1'$ and $\widehat{v}$. Now, let $P_2$ be the graph obtained from $P_1$ by inserting $\widehat{u}$ and $\widehat{v}$ between $u$ and $v$ with the edges $\{u,\widehat{u}\},\{\widehat{u},\widehat{v}\}$ and $\{\widehat{v},v\}$ (see Fig.~\ref{fig:casetwo}). Because of Property \ref{condition:GridClique1} in Definition \ref{def:GridClique}, and since $u$ and $\widehat{u}$ belong to the same cell, they are adjacent in $G$. Similarly, $v$ and $\widehat{v}$ are adjacent in $G$. Moreover, because $\{\widehat{u},\widehat{v}\}\in \ma_1(\{(i,j),(i',j')\})$, it is an edge in $G$. Thus, $P_2$ is a path (resp.~cycle) in $G$. Additionally, $V(P)\subseteq V(P_2)$, and therefore $|V(P_2)|\geq k$. The only edges that appear only in one among $P_2$ and $P$ are~as~follows.
\begin{enumerate}
\item If $\widehat{u}$ belongs to $P$ and has two neighbors in $P$, then the edges between $\widehat{u}$ and these two neighbors might belong only to $P$, and the edge between these two neighbors belongs only to $P_2$. As $\widehat{u}$ and its neighbors in $P$ belong to the same cell (by the choice of $\widehat{u}$), none of these edges is bad, and none of them crosses different cells. The same holds for $\widehat{v}$.

\item If $\widehat{u}$ belongs to $P$ and has only one neighbor in $P$, the edge between $\widehat{u}$ and this neighbor might belong only to $P$. The same holds for $\widehat{v}$.

\item $\{u,v\}\in E(P)\setminus E(P_2)$ is a bad edge that crosses different cells by its initial choice.

\item $\{u,\widehat{u}\}$ might belong only to $P_2$, and it is a neither a bad edge nor it crosses different cells because $u$ and $\widehat{u}$ belong to the same cell. The same holds for $\{v,\widehat{v}\}$.

\item $\{\widehat{u},\widehat{v}\}\in E(P_2)\setminus E(P)$ is a not a bad edge because both $\widehat{u}$ and $\widehat{v}$ are marked (since $\{\widehat{u},\widehat{v}\}\in \ma_1(\{(i,j),(i',j')\})$), but it crosses different cells.
\end{enumerate}

Thus, $P_2$ has no bad edge that does not belong to $P$, and $P$ has at least one bad edge that does not belong to $P_2$  (specifically, $\{u,v\}$), and therefore $P_2$ has fewer bad edges than $P$. Moreover, notice that the items above also imply that $P_2$ has at most one edge that crosses different cells and does not belong to $P$ (specifically, $\{\widehat{u},\widehat{v}\}$), and $P$ has at least one edge that crosses the {\em same} cells and does not belong to $P_2$ (specifically, $\{u,v\}$). Therefore, $P_2$ also has the property of $P$ that for every two distinct cells $(\widetilde{i},\widetilde{j})$ and $(\widetilde{i}',\widetilde{j}')$, there exist at most $5$ edges $\{\widetilde{u},\widetilde{v}\}\in E(P_2)$ such that $f(\widetilde{u})=(\widetilde{i},\widetilde{j})$ and $f(\widetilde{v})=(\widetilde{i}',\widetilde{j}')$. Therefore, we have reached a contradiction to the minimality of the number of bad edges in our choice of $P$.
 
\medskip
\noindent In both cases we have reached a contradiction, and therefore the proof is complete.
\end{proof}

Next, we further strengthen Lemma \ref{lem:propertyMaMain} with the following definition and Lemma \ref{lem:propertyMa}. Intuitively, the following definition says that a cell is good  with respect to some path if either none of its unmarked vertices is traversed by that path, or all of its unmarked vertices are traversed by that path consecutively and can be ``flanked'' only by marked vertices (see Fig.~\ref{fig:goodcell}).

\begin{definition}\label{def:goodCell}
Let $G$ be a clique-grid with representation $f$. Let $P$ be a path (resp.~cycle) in $G$ with endpoints $x,y$ (resp.~no endpoints). We say that a cell $(i,j)\in[t]\times[t]$ is {\em good} if {\bf (i)} $V(P)=f^{-1}(i,j)\setminus\ma^\star(i,j)$, or  {\bf (ii)} $V(P)\cap(f^{-1}(i,j)\setminus\ma^\star(i,j))=\emptyset$, or {\bf (iii)} there exist distinct $u,v\in (V(P)\cap \ma^\star(i,j))\cup(\{x,y\}\cap f^{-1}(i,j))$ (resp.~not necessarily distinct $u,v\in V(P)\cap \ma^\star(i,j)$) such that the set $I$ of internal vertices of the (resp.~a)  subpath of $P$ between $u$ and $v$ is precisely $f^{-1}(i,j)\setminus(\ma^\star(i,j)\cup\{u,v\})$;\footnote{In other words, $I\subseteq f^{-1}(i,j)\setminus\ma^\star(i,j)$ and $(f^{-1}(i,j)\setminus\ma^\star(i,j))\setminus I$ can only include endpoints of this subpath, in which case $P$ is a path and any included endpoint is an endpoint of $P$ as well.} otherwise, it is {\em bad}.
\end{definition}

\begin{figure}
\begin{center}
\begin{subfigure}[b]{0.24\textwidth}
\begin{tikzpicture}[scale=1]
\draw[green,thick] (0.5,2.5)--(2.5,2.5);
\node[red] (a1) at (1.5,2.5) {$\bullet$}; 
\node[] (a1) at (1.5,1.5) {$\bullet$}; 
\node[] (a1) at (0.5,0.5) {$\bullet$}; 
\node[] (a1) at (2.5,0.5) {$\bullet$}; 
\node[] (a1) at (1.5,0.5) {$\bullet$}; 
\node[red] (a1) at (0.5,2.5) {$\bullet$}; 
\node[red] (a1) at (2.5,2.5) {$\bullet$}; 
\draw (0,0)--(0,3)--(3,3)--(3,0)--(0,0);
\end{tikzpicture}
\end{subfigure}
\begin{subfigure}[b]{0.24\textwidth}
\begin{tikzpicture}[scale=1]
\draw[green,thick] (0.5,0.5)--(1.5,0.5)--(1.5,1.5);
\draw[green,thick] (-0.3,1) to[out=-90,in=180] (0.5,0.5);
\draw[green,thick] (1.8,3.2) to[out=-90,in=45] (1.5,1.5);
\node[red] (a1) at (1.5,2.5) {$\bullet$}; 
\node[] (a1) at (1.5,1.5) {$\bullet$}; 
\node[] (a1) at (0.5,0.5) {$\bullet$}; 
\node[] (a1) at (2.5,0.5) {$\bullet$}; 
\node[] (a1) at (1.5,0.5) {$\bullet$}; 
\node[red] (a1) at (0.5,2.5) {$\bullet$}; 
\node[red] (a1) at (2.5,2.5) {$\bullet$}; 
\draw (0,0)--(0,3)--(3,3)--(3,0)--(0,0);
\end{tikzpicture}
\end{subfigure}
\begin{subfigure}[b]{0.24\textwidth}
\begin{tikzpicture}[scale=1]
\draw[green,thick] (0.5,2.5)--(2.5,2.5)--(1.5,1.5)--(0.5,2.5);
\node[red] (a1) at (1.5,2.5) {$\bullet$}; 
\node[] (a1) at (1.5,1.5) {$\bullet$}; 
\node[] (a1) at (0.5,0.5) {$\bullet$}; 
\node[] (a1) at (2.5,0.5) {$\bullet$}; 
\node[] (a1) at (1.5,0.5) {$\bullet$}; 
\node[red] (a1) at (0.5,2.5) {$\bullet$}; 
\node[red] (a1) at (2.5,2.5) {$\bullet$}; 
\draw (0,0)--(0,3)--(3,3)--(3,0)--(0,0);
\end{tikzpicture}
\end{subfigure}
\begin{subfigure}[b]{0.24\textwidth}
\begin{tikzpicture}[scale=1]
 \draw[green,thick] (0.5,0.5)--(0.5,2.5)--(2.5,2.5);
\draw[green,thick] (-0.3,1) to[out=-90,in=180] (0.5,0.5);
\node[red] (a1) at (1.5,2.5) {$\bullet$}; 
\node[] (a1) at (1.5,1.5) {$\bullet$}; 
\node[] (a1) at (0.5,0.5) {$\bullet$}; 
\node[] (a1) at (2.5,0.5) {$\bullet$}; 
\node[] (a1) at (1.5,0.5) {$\bullet$}; 
\node[red] (a1) at (0.5,2.5) {$\bullet$}; 
\node[red] (a1) at (2.5,2.5) {$\bullet$}; 
\draw (0,0)--(0,3)--(3,3)--(3,0)--(0,0);
\end{tikzpicture}
\end{subfigure}
\end{center}
\caption{Illustration of  good cells. Vertices colored black and red are marked and unmarked vertices, respectively. The green curve represents the path/cycle $P$.}
\label{fig:goodcell}
\end{figure}

It will be convenient to have, as an intermediate step, a definition and lemma that are weaker than Definition \ref{def:goodCell} and Lemma \ref{lem:propertyMa}. Intuitively, this definition drops the requirement that none or all the unmarked vertices of a cell should be visited by the path at hand, but only requires that those unmarked vertices that are visited, are visited consecutively and can be ``flanked'' only by marked vertices (see Fig.~\ref{fig:nicecell}).

\begin{definition}\label{def:niceCell}
Let $G$ be a clique-grid with representation $f$. Let $P$ be a path (resp.~cycle) in $G$ with endpoints $x,y$ (resp.~no endpoints). We say that a cell $(i,j)\in[t]\times[t]$ is {\em nice} if {\bf (i)} $V(P)\subseteq f^{-1}(i,j)\setminus\ma^\star(i,j)$, or {\bf (ii)} $V(P)\cap(f^{-1}(i,j)\setminus\ma^\star(i,j))=\emptyset$, or {\bf (iii)}  there exist distinct $u,v\in (V(P)\cap \ma^\star(i,j))\cup(\{x,y\}\cap f^{-1}(i,j))$ (resp.~not necessarily distinct $u,v\in V(P)\cap \ma^\star(i,j)$) such that the set of internal vertices of the (resp.~a)  subpath of $P$ between $u$ and $v$ is  precisely $(V(P)\cap f^{-1}(i,j))\setminus(\ma^\star(i,j)\cup\{u,v\})$.
\end{definition}

\begin{figure}
\begin{center}
\begin{tikzpicture}[scale=1]
\draw[green,thick] (0.5,2.5)--(2.5,2.5);
\node[red] (a1) at (1.5,2.5) {$\bullet$}; 
\node[red] (a1) at (1.5,1.5) {$\bullet$}; 
\node[] (a1) at (0.5,0.5) {$\bullet$}; 
\node[] (a1) at (2.5,0.5) {$\bullet$}; 
\node[] (a1) at (1.5,0.5) {$\bullet$}; 
\node[red] (a1) at (0.5,2.5) {$\bullet$}; 
\node[red] (a1) at (2.5,2.5) {$\bullet$}; 
\draw (0,0)--(0,3)--(3,3)--(3,0)--(0,0);
\end{tikzpicture}
\end{center}
\caption{A nice cell which is not good. Vertices colored black and red are marked and unmarked vertices, respectively. The green curve represents the path $P$.}
\label{fig:nicecell}
\end{figure}

\begin{lemma}\label{lem:propertyMaHelper}
Let $G$ be a clique-grid with representation $f$ that has a path (resp.~cycle) on at least $k$ vertices. Then, $G$ also has a path (resp.~cycle) $P$ on at least $k$ vertices with the following property:  every cell $(i,j)\in[t]\times[t]$ is nice.
\end{lemma}

\begin{proof}
Given a path (resp.~cycle) $P$ with endpoints $x,y$ (resp.~no endpoints) and a cell $(i,j)\in[t]\times[t]$, we say that a subpath of $P$ is {\em $(i,j)$-nice} if it has distinct endpoints $u,v\in (V(P)\cap \ma^\star(i,j))\cup(\{x,y\}\cap f^{-1}(i,j))$ (resp.~$u,v\in V(P)\cap \ma^\star(i,j)$) and its set of internal vertices is a subset  $I$ of $f^{-1}(i,j)\setminus \ma^\star(i,j)$ such that if this subset $I$ is empty, then the subpath has an endpoint in $f^{-1}(i,j)\setminus \ma^\star(i,j)$ (which implies that $P$ is a path and $\{u,v\}\cap\{x,y\}\cap(f^{-1}(i,j)\setminus \ma^\star(i,j))\neq\emptyset$); we further say that a subpath of $P$ is {\em nice} if it is $(i,j)$-nice for some $(i,j)$. By Lemma \ref{lem:propertyMaMain}, $G$ has a path (resp.~cycle) on at least $k$ vertices with the following property: every edge $\{u,v\}$ of that path where $f(u)\neq f(v)$ is good. Among all such paths (resp.~cycles), let $P$ be one with minimum number of nice subpaths, and let $x,y$ be its endpoints (resp.~no endpoints). (Notice that if $x$ is unmarked, then because every edge $\{u,v\}$ of $P$ where $f(u)\neq f(v)$ is good, it must be that $x$ is an endpoint of a nice subpath. The same holds for $y$.) We next show that for every cell $(i,j)\in[t]\times[t]$, $P$ has at most one nice $(i,j)$-subpath. Because either $V(P)\subseteq f^{-1}(i,j)$ or every vertex in $(V(P)\cap f^{-1}(i,j))\setminus(\ma^\star(i,j)\cup\{x,y\})$ (resp.~$(V(P)\cap f^{-1}(i,j))\setminus\ma^\star(i,j)$) must be an internal vertex of a nice subpath (since every edge $\{u,v\}$ of $P$ where $f(u)\neq f(v)$ is good), this would imply that every cell $(i,j)\in[t]\times[t]$ is nice, which will complete the proof. Targeting a contradiction, suppose that $P$ yields some cell $(i,j)$ such that there exist two distinct subpaths $Q,Q'$ of $P$ that are $(i,j)$-nice (see Fig.~\ref{fig:proofof13}), that is, each of them has both endpoints in $\ma^\star(i,j)\cup(\{x,y\}\cap f^{-1}(i,j))$ (resp.~$\ma^\star(i,j)$) and the set of its internal vertices is a subset of $f^{-1}(i,j)\setminus \ma^\star(i,j)$ that is either non-empty or some endpoint belongs to $\{x,y\}
\cap (f^{-1}(i,j)\setminus \ma^\star(i,j))$.

Note that if $Q$ and $Q'$ intersect, then they intersect only at their endpoints. Define $\widehat{P}$ by removing from $P$ all the internal vertices of $Q'$ as well as its endpoint in $f^{-1}(i,j)\setminus\ma^\star(i,j)$ if such an endpoint exists (in which case $P$ is a path and this endpoint it is also an endpoint of $P$), and inserting them arbitrarily between the vertices of $Q$ (where multiple vertices can be inserted between two vertices); see Fig.~\ref{fig:proofof13}. By Property \ref{condition:GridClique1} in Definition \ref{def:GridClique}, we have that $\widehat{P}$ is also a path (resp.~cycle).
Clearly, $|V(\widehat{P})|=|V(P)|\geq k$, and it is also directly implied by the construction that $\widehat{P}$ also has the property that every edge $\{u,v\}\in E(\widehat{P})$ where $f(u)\neq f(v)$ is good (since we did not make any change with respect to the set of edges that cross different cells). Notice that each subpath that is nice with respect to $\widehat{P}$ is either the subpath obtained by merging $Q$ and $Q'$ or a subpath that also exists in $P$ and is therefore also a nice subpath with respect to $P$.  Therefore, $
\widehat{P}$ has one fewer nice subpath than $P$, which contradicts the minimality of $P$.
\end{proof}

\begin{figure}
\begin{center}
\begin{subfigure}[b]{0.24\textwidth}
\begin{tikzpicture}[scale=1]
\path [draw=green,snake it,thick] (0.5,1.5)--(-0.25,1.5);
\path [draw=green,snake it,thick] (0.5,0.5)--(-0.25,0.5);
\path [draw=green,snake it,thick] (2.5,1.5)--(3,2);
\path [draw=green,snake it,thick] (2.5,0.5)--(3,1);
\draw[green,thick] (0.5,1.5)--(0.5,2.5)--(2.5,2.5)--(2.5,1.5);
\draw[green,thick] (0.5,0.5)--(1.2,1.5)--(1.8,1.5)--(2.5,0.5);
\node[red] (a1) at (1.5,2.5) {$\bullet$}; 
\node[red] (a1) at (1.2,1.5) {$\bullet$}; 
\node[red] (a1) at (1.8,1.5) {$\bullet$}; 
\node[] (a1) at (0.5,0.5) {$\bullet$}; 
\node[] (a1) at (2.5,0.5) {$\bullet$}; 
\node[] (a1) at (1.5,0.5) {$\bullet$}; 
\node[red] (a1) at (0.5,2.5) {$\bullet$}; 
\node[red] (a1) at (2.5,2.5) {$\bullet$}; 
\node[] (a1) at (0.5,1.5) {$\bullet$}; 
\node[] (a1) at (2.5,1.5) {$\bullet$}; 
\node (a1) at (1.1,2.7) {$Q$}; 
\node (a1) at (1.3,1.2) {$Q'$}; 
\draw (0,0)--(0,3)--(3,3)--(3,0)--(0,0);
\end{tikzpicture}
\label{13a}
\end{subfigure}~
\begin{subfigure}[b]{0.24\textwidth}
\begin{tikzpicture}[scale=1]
\draw[dashed,thick] (0.5,1.5)--(-0.2,1.5);
\draw[dashed,thick] (0.5,0.5)--(-0.2,0.5);
\draw[dashed,thick] (2.5,1.5)--(3.25,1.8);
\draw[dashed,thick] (2.5,0.5)--(3.25,0.8);
\draw[dashed,thick] (0.5,1.5)--(1.8,1.5)--(0.5,2.5)--(2.5,2.5)--(2.5,1.5);
\draw[dashed,thick] (0.5,0.5)to [out=-30,in=-150] (2.5,0.5);
\node[red] (a1) at (1.5,2.5) {$\bullet$}; 
\node[red] (a1) at (1.2,1.5) {$\bullet$}; 
\node[red] (a1) at (1.8,1.5) {$\bullet$}; 
\node[] (a1) at (0.5,0.5) {$\bullet$}; 
\node[] (a1) at (2.5,0.5) {$\bullet$}; 
\node[] (a1) at (1.5,0.5) {$\bullet$}; 
\node[red] (a1) at (0.5,2.5) {$\bullet$}; 
\node[red] (a1) at (2.5,2.5) {$\bullet$}; 
\node[] (a1) at (0.5,1.5) {$\bullet$}; 
\node[] (a1) at (2.5,1.5) {$\bullet$}; 
\draw (0,0)--(0,3)--(3,3)--(3,0)--(0,0);
\end{tikzpicture}
\end{subfigure}~
\begin{subfigure}[b]{0.24\textwidth}
\begin{tikzpicture}[scale=1]
\path [draw=green,snake it,thick] (0.5,0.5)--(-0.25,0.5);
\path [draw=green,snake it,thick] (0.5,1.5)--(-0.25,1.5);
\draw[green,thick] (0.5,1.5)--(0.5,2.5)--(2.5,2.5);
\draw[green,thick] (0.5,0.5)--(1.2,1.5)--(1.8,1.5);
\node[] (a1) at (2.75,2.5) {$x$};
\node[] (a1) at (2,1.5) {$y$};
\node[red] (a1) at (1.5,2.5) {$\bullet$}; 
\node[red] (a1) at (1.2,1.5) {$\bullet$}; 
\node[red] (a1) at (1.8,1.5) {$\bullet$}; 
\node[] (a1) at (0.5,0.5) {$\bullet$}; 
\node[] (a1) at (2.5,0.5) {$\bullet$}; 
\node[] (a1) at (1.5,0.5) {$\bullet$}; 
\node[red] (a1) at (0.5,2.5) {$\bullet$}; 
\node[red] (a1) at (2.5,2.5) {$\bullet$}; 
\node[] (a1) at (0.5,1.5) {$\bullet$}; 
\node[] (a1) at (2.5,1.5) {$\bullet$}; 
\node (a1) at (1.1,2.7) {$Q$}; 
\node (a1) at (1.3,1.2) {$Q'$}; 
\draw (0,0)--(0,3)--(3,3)--(3,0)--(0,0);
\end{tikzpicture}
\label{13b}
\end{subfigure}~
\begin{subfigure}[b]{0.24\textwidth}
\begin{tikzpicture}[scale=1]
\draw[dashed,thick] (0.5,1.5)--(1.8,1.5)--(0.5,2.5)--(2.5,2.5);
\draw[dashed,thick] (0.5,0.5)--(-0.25,0.5);
\draw[dashed,thick] (0.5,1.5)--(-0.25,1.5);
\node[] (a1) at (2.75,2.5) {$x$};
\node[] (a1) at (2,1.5) {$y$};
\node[red] (a1) at (1.5,2.5) {$\bullet$}; 
\node[red] (a1) at (1.2,1.5) {$\bullet$}; 
\node[red] (a1) at (1.8,1.5) {$\bullet$}; 
\node[] (a1) at (0.5,0.5) {$\bullet$}; 
\node[] (a1) at (2.5,0.5) {$\bullet$}; 
\node[] (a1) at (1.5,0.5) {$\bullet$}; 
\node[red] (a1) at (0.5,2.5) {$\bullet$}; 
\node[red] (a1) at (2.5,2.5) {$\bullet$}; 
\node[] (a1) at (0.5,1.5) {$\bullet$}; 
\node[] (a1) at (2.5,1.5) {$\bullet$}; 
\draw (0,0)--(0,3)--(3,3)--(3,0)--(0,0);
\end{tikzpicture}
\end{subfigure}
\end{center}
\caption{The proof of Lemma~\ref{lem:propertyMaHelper}. Vertices colored black and red are the marked and unmarked vertices in the cell, respectively. In the first figure the union of internal vertices of $Q$ and $Q'$ is the set of unmarked vertices in the cell, and the second figure depicts how to reroute to make the cell nice. The third figure illustrate the case when both the endpoints $x$ and $y$ of the path $P$ are in the cell,  and the fourth figure depicts how to reroute to make the cell nice.}
\label{fig:proofof13}
\end{figure}
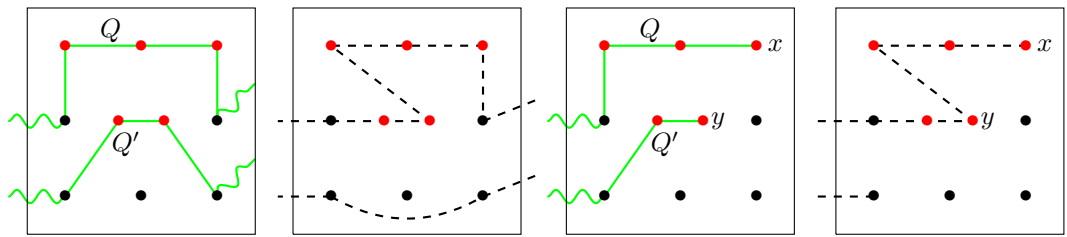

We now prove the main lemma of this section.

\begin{lemma}\label{lem:propertyMa}
Let $G$ be a clique-grid with representation $f$ that has a path (resp.~cycle) on at least $k$ vertices. Then, $G$ also has a path (resp.~cycle) $P$ on at least $k$ vertices with the following property:  every cell $(i,j)\in[t]\times[t]$ is good.
\end{lemma}

\begin{proof}
By Lemma \ref{lem:propertyMaHelper}, $G$ has a path (resp.~cycle) $P$ on at least $k$ vertices with the following property:  every cell $(i,j)\in[t]\times[t]$ is nice. Among all such paths (resp.~cycles), let $P$ be one with minimum number of bad cells.  Next, we show that $P$ yields no bad cell, which will complete the proof. Targeting a contradiction, suppose that $P$ yields some bad cell $(i,j)\in[t]\times[t]$. Because this cell is not good, $(V(P)\cap f^{-1}(i,j))\setminus\ma^\star(i,j)\neq\emptyset$. Further, because $(i,j)$ is nice, either $V(P)\subseteq f^{-1}(i,j)\setminus 
\ma^\star(i,j)$ or there exist distinct $u,v\in (V(P)\cap \ma^\star(i,j))\cup(\{x,y\})\cap f^{-1}(i,j))$ (resp.~$u,v\in V(P)\cap \ma^\star(i,j)$) such that the set of internal vertices of the (resp.~a)  subpath $Q$ of $P$ between $u$ and $v$ is precisely $(V(P)\cap f^{-1}(i,j))\setminus(\ma^\star(i,j)\cup\{u,v\})$ (see Fig.~\ref{fig:proofof14}).  In the first case, notice that since $f^{-1}(i,j)\setminus\ma^\star(i,j)$ induces a clique (by Property \ref{condition:GridClique1} in Definition \ref{def:GridClique}) and its size is at least $k$ (because $|V(P)|\geq k$), it is clear that $G$ contains a path (resp.~cycle) whose vertex set is $f^{-1}(i,j)\setminus\ma^\star(i,j)$ and which has at least $k$ vertices, for which every cell is trivially good. Thus, we next suppose that only the second case happens.

 Notice that $Q$ must contain a vertex from  $\ma^\star(i,j)$ as an endpoint, because its endpoints $u,v\in (V(P)\cap \ma^\star(i,j))\cup(\{x,y\})\cap f^{-1}(i,j))$ (resp.~$u,v\in V(P)\cap \ma^\star(i,j)$) and it is not possible that $\{u,v\}=\{x,y\}$ (since then the first case happens). Because also $(V(P)\cap f^{-1}(i,j))\setminus\ma^\star(i,j)\neq\emptyset$, we know that $Q$ contains one edge $\{a,b\}$ with both endpoints from $f^{-1}(i,j)$. Then, we derive $\widehat{P}$ from $P$ by inserting all the vertices in $(f^{-1}(i,j)\setminus \ma^\star(i,j))\setminus V(P)$ between $a$ and $b$ in some arbitrary order (see Fig.~\ref{fig:proofof14}). By Property \ref{condition:GridClique1} in Definition \ref{def:GridClique}, we still have a path (resp.~cycle). Further, notice that $(i,j)$ is a good cell with respect to $\widehat{P}$. As the adjacencies of all vertices outside the cell $(i,j)$ are the same in $P$ and $\widehat{P}$, we have that $\widehat{P}$ has only nice cells (because $P$ has this property), and that every cell that is bad with respect to $\widehat{P}$ is also bad with respect to $P$. Thus, we obtain a path (resp.~cycle) on at least $k$ vertices with fewer bad cells than $P$ and still with the property every cell $(i,j)\in[t]\times[t]$ is nice. This is a contradiction to the choice of $P$, and therefore the proof is complete.
\end{proof}

\begin{figure}
\begin{center}
\begin{subfigure}[b]{0.24\textwidth}
\begin{tikzpicture}[scale=1]
\draw[green,thick] (1.8,1.5)--(0.5,2.5)--(2.5,2.5);
\node[] (a1) at (2.75,2.5) {$x$};
\node[] (a1) at (2,1.5) {$y$};
\node[red] (a1) at (1.5,2.5) {$\bullet$}; 
\node[red] (a1) at (1.2,1.5) {$\bullet$}; 
\node[red] (a1) at (1.8,1.5) {$\bullet$}; 
\node[] (a1) at (0.5,0.5) {$\bullet$}; 
\node[] (a1) at (2.5,0.5) {$\bullet$}; 
\node[] (a1) at (1.5,0.5) {$\bullet$}; 
\node[red] (a1) at (0.5,2.5) {$\bullet$}; 
\node[red] (a1) at (2.5,2.5) {$\bullet$}; 
\node[] (a1) at (0.5,1.5) {$\bullet$}; 
\node[] (a1) at (2.5,1.5) {$\bullet$}; 
\node (a1) at (1.1,2.7) {$P$}; 
\draw (0,0)--(0,3)--(3,3)--(3,0)--(0,0);
\end{tikzpicture}
\end{subfigure}
\begin{subfigure}[b]{0.24\textwidth}
\begin{tikzpicture}[scale=1]
\draw[dashed,thick] (1.8,1.5)--(1.2,1.5)--(0.5,2.5)--(2.5,2.5);
\node[] (a1) at (2.75,2.5) {$x$};
\node[] (a1) at (2,1.5) {$y$};
\node[red] (a1) at (1.5,2.5) {$\bullet$}; 
\node[red] (a1) at (1.2,1.5) {$\bullet$}; 
\node[red] (a1) at (1.8,1.5) {$\bullet$}; 
\node[] (a1) at (0.5,0.5) {$\bullet$}; 
\node[] (a1) at (2.5,0.5) {$\bullet$}; 
\node[] (a1) at (1.5,0.5) {$\bullet$}; 
\node[red] (a1) at (0.5,2.5) {$\bullet$}; 
\node[red] (a1) at (2.5,2.5) {$\bullet$}; 
\node[] (a1) at (0.5,1.5) {$\bullet$}; 
\node[] (a1) at (2.5,1.5) {$\bullet$}; 
\draw (0,0)--(0,3)--(3,3)--(3,0)--(0,0);
\end{tikzpicture}
\end{subfigure}
\begin{subfigure}[b]{0.24\textwidth}
\begin{tikzpicture}[scale=1]
\draw[green,thick] (0.5,1.5)--(0.5,2.5)--(2.5,2.5)--(2.5,1.5);
\path [draw=green,snake it,thick] (0.5,1.5)--(-0.25,1.2);
\path [draw=green,snake it,thick] (2.5,1.5)--(3.25,1.8);
\node[red] (a1) at (1.5,2.5) {$\bullet$}; 
\node[red] (a1) at (1.2,1.5) {$\bullet$}; 
\node[red] (a1) at (1.8,1.5) {$\bullet$}; 
\node[] (a1) at (0.5,0.5) {$\bullet$}; 
\node[] (a1) at (2.5,0.5) {$\bullet$}; 
\node[] (a1) at (1.5,0.5) {$\bullet$}; 
\node[red] (a1) at (0.5,2.5) {$\bullet$}; 
\node[red] (a1) at (2.5,2.5) {$\bullet$}; 
\node[] (a1) at (0.5,1.5) {$\bullet$}; 
\node[] (a1) at (2.5,1.5) {$\bullet$}; 
\node (a1) at (1.1,2.7) {$Q$}; 
\draw (0,0)--(0,3)--(3,3)--(3,0)--(0,0);
\end{tikzpicture}
\end{subfigure}
\begin{subfigure}[b]{0.24\textwidth}
\begin{tikzpicture}[scale=1]
\draw[dashed,thick] (0.5,1.5)--(0.5,2.5)--(1.2,1.5)--(1.8,1.5)--(1.5,2.5)--(2.5,2.5)--(2.5,1.5);
\draw[dashed,thick] (0.5,1.5)--(-0.25,1.2);
\draw[dashed,thick] (2.5,1.5)--(3.25,1.8);
\node[red] (a1) at (1.5,2.5) {$\bullet$}; 
\node[red] (a1) at (1.2,1.5) {$\bullet$}; 
\node[red] (a1) at (1.8,1.5) {$\bullet$}; 
\node[] (a1) at (0.5,0.5) {$\bullet$}; 
\node[] (a1) at (2.5,0.5) {$\bullet$}; 
\node[] (a1) at (1.5,0.5) {$\bullet$}; 
\node[red] (a1) at (0.5,2.5) {$\bullet$}; 
\node[red] (a1) at (2.5,2.5) {$\bullet$}; 
\node[] (a1) at (0.5,1.5) {$\bullet$}; 
\node[] (a1) at (2.5,1.5) {$\bullet$}; 
\draw (0,0)--(0,3)--(3,3)--(3,0)--(0,0);
\end{tikzpicture}
\end{subfigure}
\end{center}
\caption{The proof of Lemma~\ref{lem:propertyMa}. Vertices colored black and red are marked and unmarked vertices, respectively. In the first figure $V(P)\subseteq f^{-1}(i,j)\setminus 
\ma^\star(i,j)$, and the second figure illustrates that there is a path of length at least $\vert V(P)\vert$ whose vertex set is the set of unmarked vertices in the cell. The third figure illustrates the case where $P$  is not fully contained  the cell,  and the fourth figure depicts a possibility to reroute it to make the cell good.}
\label{fig:proofof14}
\end{figure}
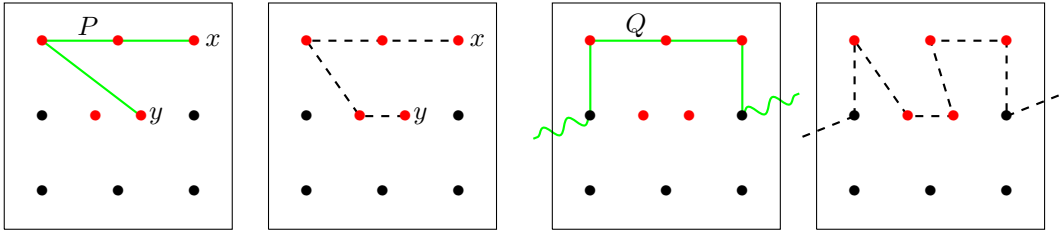

\section{The Algorithm}\label{sec:alg}

Our algorithm is based on a reduction of {\sc Long Path} (resp.~{\sc Long Cycle}) on unit disk graphs to the weighted version of the problem, called {\sc Weighted Long Path} (resp.~{\sc Weighted Long Cycle}), on unit disk graphs of treewidth $\OO(\sqrt{k})$. In {\sc Weighted Long Path} (resp.~{\sc Weighted Long Cycle}), we are given a graph $G$ with a weight function $w: V(G)\rightarrow\mathbb{N}$ and an integer $k\in\mathbb{N}$, and the objective is to determine whether $G$ has a path (resp.~cycle) whose weight, defined as the sum of the weights of its vertices, is at least $k$.

The following proposition will be immediately used in our algorithm.

\begin{proposition}[\cite{DBLP:journals/iandc/BodlaenderCKN15,DBLP:journals/jacm/FominLPS16}]\label{prop:weightedPath}
{\sc Weighted Long Path} and {\sc Weighted Long Cycle} are solvable in time $2^{\OO(\tw)}n$ where $\tw$ is the treewidth of the input graph.
\end{proposition}

\subparagraph*{Algorithm Specification.} We call our algorithm \alg.  Given an instance $(G,k)$ of {\sc Long Path} (resp.~{\sc Long Cycle}) on unit disk graphs, it works as follows.
\begin{enumerate}
\item\label{alg1} Use Proposition \ref{prop:cliqueGrid} to obtain a representation $f: V(G)\rightarrow [t]\times[t]$ of $G$.

\item\label{alg2} Use Observation \ref{obs:sizeMa} to compute $\ma^\star(i,j)$ for every cell $(i,j)\in[t]\times[t]$. 
Let $\ma^\star=\bigcup_{(i,j)\in [t]\times [t]}\ma^\star(i,j)$. 

\item\label{alg3} Let $G'$ be the graph  defined as follows (see Fig.~\ref{fig:alggraphs}). For any cell $(i,j)\in[t]\times[t]$, let $c_{(i,j)}$ denote a vertex in $f^{-1}(i,j)\setminus \ma^\star(i,j)$ (chosen arbitrarily), where if no such vertex exists, let $c_{(i,j)}=\nil$. Then, $V(G')=\ma^\star\cup(\{c_{(i,j)}: (i,j)\in[t]\times[t]\}\setminus\{\nil\})$ and $E(G')=E(G[V(G')])$. Because $G'$ is an induced subgraph of $G$, it is a unit disk graph.

\item\label{alg4} Define $w:V(G')\rightarrow\mathbb{N}$ as follows. For every $v\in V(G')$, if $v=c_{(i,j)}$ for some $(i,j)\in[t]\times[t]$ then $w(v)=|f^{-1}(i,j)\setminus \ma^\star(i,j)|$, and otherwise $w(v)=1$.

\item\label{alg6b} Let $G^\star$ be the graph defined as follows (see Fig.~\ref{fig:alggraphs}): $V(G^\star)=V(G')$ and $E(G^\star)=E(G')\setminus\{\{c_{(i,j)},v\}\in E(G'): (i,j)\in[t]\times[t], v\notin f^{-1}(i,j)\}$.

\item\label{alg5} Let $\Delta$ be the maximum degree of $G^{\star}$. Use Proposition \ref{prop:treewidth} to decide either $\tw(G^{\star})>100\Delta^3\sqrt{2k}$ or $\tw(G^{\star})\leq 500\Delta^3\sqrt{2k}$.

\item\label{alg6} If it was decided that $\tw(G^{\star})>100\Delta^3\sqrt{2k}$, then return \yes\ and terminate.

\item\label{alg7} Use Proposition \ref{prop:weightedPath} to determine whether $(G^\star,w,k)$ is a \yes-instance of {\sc Weighted Long Path} (resp.~{\sc Weighted Long Cycle}). If the answer is positive, then return \yes, and otherwise return \no.
\end{enumerate}

\begin{figure}
\begin{center}
\begin{subfigure}[b]{0.49\textwidth}
\begin{tikzpicture}[scale=1]
\draw (0,1)--(6,1);
\draw (0,2)--(6,2);
\draw (0,3)--(6,3);
\draw (0,4)--(6,4);
\draw (0,5)--(6,5);
\draw (0,1)--(0,5);
\draw (1,1)--(1,5);
\draw (2,1)--(2,5);
\draw (3,1)--(3,5);
\draw (4,1)--(4,5);
\draw (5,1)--(5,5);
\draw (6,1)--(6,5);
\draw (2.8,2.8)--(2.7,2.3)--(2.2,2.85)--(2.1,2.1)--(2.8,2.8)--(2.2,2.85);
\draw (2.1,2.1)--(2.7,2.3);
\node[] (a1) at (2.8,2.8) {$\bullet$}; 
\node[] (a1) at (2.7,2.3) {$\bullet$}; 
\node[] (a1) at (2.2,2.85) {$\bullet$}; 
\node[] (a1) at (2.1,2.1) {$\bullet$}; 
\draw[blue] (1.7,2.3)--(2.2,2.85)--(1.2,2.85);
\draw[blue]  (1.7,2.3)--(2.1,2.1)to [out=90,in=-20] (1.2,2.85);
\draw[blue]  (1.2,2.1)--(2.1,2.1)--(0.9,1.8);
\draw (1.2,2.1)--(1.7,2.3)--(1.2,2.85)--(1.2,2.1);
\node[] (a1) at (1.2,2.1) {$\bullet$}; 
\node[] (a1) at (1.7,2.3) {$\bullet$}; 
\node[] (a1) at (1.2,2.85) {$\bullet$}; 
\draw[blue] (1.2,2.85)--(0.9,1.8);
\draw[blue]  (0.1,1.85)--(1.2,2.1)--(0.9,1.8);
\draw[blue]  (1.2,2.1)--(0.8,1.1);
\draw (0.8,1.1) --(0.9,1.8) -- (0.1,1.85) -- (0.8,1.1);
\node[] (a1) at (0.8,1.1) {$\bullet$}; 
\node[] (a1) at (0.9,1.8) {$\bullet$}; 
\node[] (a1) at (0.1,1.85) {$\bullet$}; 
\draw (3.8,3.1)--(3.7,3.7);
\node[] (a1) at (3.8,3.1) {$\bullet$}; 
\node[] (a1) at (3.7,3.7) {$\bullet$}; 
\draw[blue] (2.8,2.8)--(3.8,3.1)--(2.7,2.3);
\draw[blue] (2.8,2.8)--(3.7,3.7)--(2.7,2.3);
\draw (4.7,2.1) --(4.4,2.8)--(4.1,2.85)--(4.7,2.1);
\node[] (a1) at (4.7,2.1) {$\bullet$}; 
\node[] (a1) at (4.4,2.8) {$\bullet$}; 
\node[] (a1) at (4.1,2.85) {$\bullet$}; 
\draw[blue] (3.8,3.1)--(4.1,2.85) -- (3.7,3.7);
\draw[blue] (3.8,3.1)--(4.4,2.8)-- (3.7,3.7);
\draw (5.7,4.9) -- (5.1,4.1);
\node[] (a1) at (5.85,4.8) {$y$}; 
\node[red] (a1) at (5.7,4.9) {$\bullet$}; 
\node[] (a1) at (5.1,4.1) {$\bullet$}; 
\draw[blue] (5.1,4.1)--(4.4,2.8);
\draw[blue] (5.1,4.1)--(3.7,3.7);
\draw[blue] (4.7,2.1)-- (5.1,1.5);
\draw[blue] (4.7,2.1)-- (5.3,1.7);
\draw[blue] (4.7,2.1)-- (5.6,1.7);
\draw[blue] (4.7,2.1)-- (5.8,1.5);
\draw[blue] (4.7,2.1)-- (5.6,1.2);
\draw[red] (4.7,2.1)-- (5.2,1.1);
\draw (5.1,1.5) --(5.3,1.7)--(5.6,1.7)--(5.8,1.5)--(5.6,1.2)--(5.2,1.1)-- (5.1,1.5)--(5.6,1.7)--(5.6,1.2)-- (5.1,1.5)--(5.8,1.5)--(5.2,1.1)--(5.3,1.7)--(5.8,1.5);
\draw (5.3,1.7)--(5.6,1.2);
\draw (5.2,1.1)--(5.6,1.7);
\draw[blue](5.9,2.8)--(5.6,1.7);
\node[] (a1) at (5.1,1.5) {$\bullet$}; 
\node[] (a1) at (5.3,1.7) {$\bullet$}; 
\node[] (a1) at (5.6,1.7) {$\bullet$}; 
\node[] (a1) at (5.8,1.5) {$\bullet$}; 
\node[] (a1) at (5.6,1.2) {$\bullet$}; 
\node[red] (a1) at (5.2,1.1) {$\bullet$}; 
\node[] (a1) at (5.35,1) {$x$}; 
\node[] (a1) at (5.9,2.8) {$\bullet$}; 
\end{tikzpicture}
\end{subfigure}
\begin{subfigure}[b]{0.49\textwidth}
\begin{tikzpicture}[scale=1]
\draw (0,1)--(6,1);
\draw (0,2)--(6,2);
\draw (0,3)--(6,3);
\draw (0,4)--(6,4);
\draw (0,5)--(6,5);
\draw (0,1)--(0,5);
\draw (1,1)--(1,5);
\draw (2,1)--(2,5);
\draw (3,1)--(3,5);
\draw (4,1)--(4,5);
\draw (5,1)--(5,5);
\draw (6,1)--(6,5);
\draw (2.8,2.8)--(2.7,2.3)--(2.2,2.85)--(2.1,2.1)--(2.8,2.8)--(2.2,2.85);
\draw (2.1,2.1)--(2.7,2.3);
\node[] (a1) at (2.8,2.8) {$\bullet$}; 
\node[] (a1) at (2.7,2.3) {$\bullet$}; 
\node[] (a1) at (2.2,2.85) {$\bullet$}; 
\node[] (a1) at (2.1,2.1) {$\bullet$}; 
\draw[blue] (1.7,2.3)--(2.2,2.85)--(1.2,2.85);
\draw[blue]  (1.7,2.3)--(2.1,2.1)to [out=90,in=-20] (1.2,2.85);
\draw[blue]  (1.2,2.1)--(2.1,2.1)--(0.9,1.8);
\draw (1.2,2.1)--(1.7,2.3)--(1.2,2.85)--(1.2,2.1);
\node[] (a1) at (1.2,2.1) {$\bullet$}; 
\node[] (a1) at (1.7,2.3) {$\bullet$}; 
\node[] (a1) at (1.2,2.85) {$\bullet$}; 
\draw[blue] (1.2,2.85)--(0.9,1.8);
\draw[blue]  (0.1,1.85)--(1.2,2.1)--(0.9,1.8);
\draw[blue]  (1.2,2.1)--(0.8,1.1);
\draw (0.8,1.1) --(0.9,1.8) -- (0.1,1.85) -- (0.8,1.1);
\node[] (a1) at (0.8,1.1) {$\bullet$}; 
\node[] (a1) at (0.9,1.8) {$\bullet$}; 
\node[] (a1) at (0.1,1.85) {$\bullet$}; 
\draw (3.8,3.1)--(3.7,3.7);
\node[] (a1) at (3.8,3.1) {$\bullet$}; 
\node[] (a1) at (3.7,3.7) {$\bullet$}; 
\draw[blue] (2.8,2.8)--(3.8,3.1)--(2.7,2.3);
\draw[blue] (2.8,2.8)--(3.7,3.7)--(2.7,2.3);
\draw (4.7,2.1) --(4.4,2.8)--(4.1,2.85)--(4.7,2.1);
\node[] (a1) at (4.7,2.1) {$\bullet$}; 
\node[] (a1) at (4.4,2.8) {$\bullet$}; 
\node[] (a1) at (4.1,2.85) {$\bullet$}; 
\draw[blue] (3.8,3.1)--(4.1,2.85) -- (3.7,3.7);
\draw[blue] (3.8,3.1)--(4.4,2.8)-- (3.7,3.7);
\draw (5.7,4.9) -- (5.1,4.1);
\node[red] (a1) at (5.7,4.9) {$\bullet$}; 
\node[] (a1) at (5.1,4.1) {$\bullet$}; 
\draw[blue] (5.1,4.1)--(4.4,2.8);
\draw[blue] (5.1,4.1)--(3.7,3.7);
\draw[blue] (4.7,2.1)-- (5.1,1.5);
\draw[blue] (4.7,2.1)-- (5.3,1.7);
\draw[blue] (4.7,2.1)-- (5.6,1.7);
\draw[blue] (4.7,2.1)-- (5.8,1.5);
\draw[blue] (4.7,2.1)-- (5.6,1.2);
\draw (5.1,1.5) --(5.3,1.7)--(5.6,1.7)--(5.8,1.5)--(5.6,1.2)--(5.2,1.1)-- (5.1,1.5)--(5.6,1.7)--(5.6,1.2)-- (5.1,1.5)--(5.8,1.5)--(5.2,1.1)--(5.3,1.7)--(5.8,1.5);
\draw (5.3,1.7)--(5.6,1.2);
\draw (5.2,1.1)--(5.6,1.7);
\draw[blue](5.9,2.8)--(5.6,1.7);
\node[] (a1) at (5.1,1.5) {$\bullet$}; 
\node[] (a1) at (5.3,1.7) {$\bullet$}; 
\node[] (a1) at (5.6,1.7) {$\bullet$}; 
\node[] (a1) at (5.8,1.5) {$\bullet$}; 
\node[] (a1) at (5.6,1.2) {$\bullet$}; 
\node[red] (a1) at (5.2,1.1) {$\bullet$}; 
\node[] (a1) at (5.85,4.8) {$y$}; 
\node[] (a1) at (5.35,1) {$x$}; 
\node[] (a1) at (5.9,2.8) {$\bullet$}; 
\end{tikzpicture}
\end{subfigure}
\end{center}
\caption{The graphs $G'$ and $G^{\star}$ constructed from the graph $G$ in Figure~\ref{fig:cliquegridgraph} are depicted on the left side and right side figures, respectively. Here, $w(x)=2$, $w(y)=3$, and for all $z\in V(G')\setminus \{x,y\}$, $w(z)=1$.}
\label{fig:alggraphs}
\end{figure}
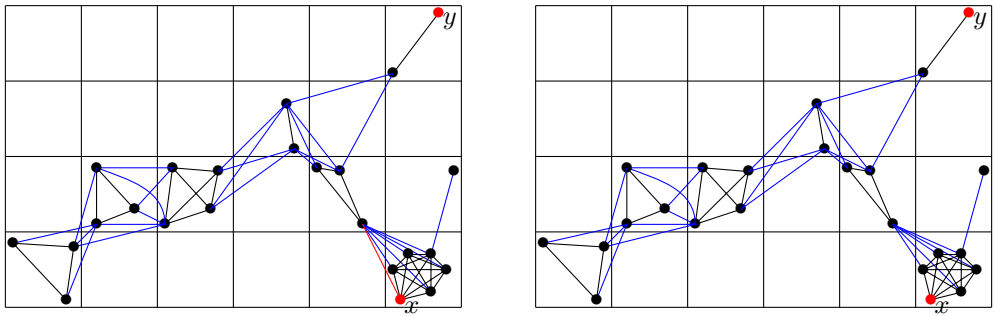

\subparagraph*{Analysis.} We first analyze the running time of the algorithm.

\begin{lemma}\label{lem:runtime}
The time complexity of \alg\ is upper bounded by $2^{\OO(\sqrt{k})}(n+m)$.
\end{lemma}

\begin{proof}
By Proposition \ref{prop:cliqueGrid} and Observation \ref{obs:sizeMa}, Steps \ref{alg1} and \ref{alg2} are performed in time $\OO(n+m)$. By the definition of $G'$, $w$ and $G^\star$, they can clearly be computed in time $\OO(n+m)$ as well (Steps \ref{alg3}, \ref{alg4} and \ref{alg6b}). Moreover, Step \ref{alg6} is done in time $\OO(1)$. By Proposition \ref{prop:treewidth}, Step \ref{alg5} is performed in time $2^{\OO(100\Delta^3\sqrt{2k})}n=2^{\OO(\Delta^3\sqrt{k})}n$. 
Thus, because we reach Step \ref{alg7} only if we do not terminate in Step \ref{alg6}, we have that by Proposition \ref{prop:weightedPath}, Step \ref{alg7} is performed in time $2^{\OO(\tw(G^\star))}n =2^{\OO(500\Delta^3\sqrt{2k})}=2^{\OO(\Delta^3\sqrt{k})}n$.

Thus, to conclude the proof, it remains to show that $\Delta=\OO(1)$. Let $\Delta'$ be the maximum degree of $G'$. 
Since $G^{\star}$ is a subgraph of $G'$, $\Delta\leq \Delta'$. Thus, to prove $\Delta=\OO(1)$, it is enough to prove that $\Delta'=\OO(1)$. To this end, let $M=\max_{(i,j)\in[t]\times[t]}|(f^{-1}(i,j)\cap V(G'))\cup(\{c_{(i,j)}\}\setminus\{\nil\})|$. Since $G'$ is a clique-grid, by Property  \ref{condition:GridClique2} in Definition \ref{def:GridClique}, we have that $\Delta'\leq M^{25}$, hence it suffices to show that $M=\OO(1)$. The definition of $G'$ yields that $M\leq \max_{(i,j)\in[t]\times[t]}|\ma^\star(i,j)|+1$. By Observation \ref{obs:sizeMa}, $\max_{(i,j)\in[t]\times[t]}|\ma^\star(i,j)|=\OO(1)$, and therefore indeed $M=\OO(1)$.
\end{proof}

Finally, we prove that the algorithm is correct.

\begin{lemma}\label{lem:correctness}
\alg\ solves {\sc Long Path} and {\sc Long Cycle} on unit disk graphs correctly.
\end{lemma}

\begin{proof}
Let $(G,k)$ be an instance of {\sc Long Path} or {\sc Long Cycle} on unit disk graphs. By the specification of the algorithm, to prove that it solves $(G,k)$ correctly, it suffices to prove that the two following conditions are satisfied.
\begin{enumerate}
\item If $\tw(G^{\star})>100\Delta^3\sqrt{2k}$, then $(G,k)$ is a \yes-instance\ of {\sc Long Path} and {\sc Long Cycle}.
\item $(G,k)$ is a \yes-instance of {\sc Long Path} (resp.~{\sc Long Cycle}) if and only if $(G^\star,w,k)$ is a \yes-instance of {\sc Weighted Long Path} (resp.~{\sc Weighted Long Cycle}).
\end{enumerate}

For the proof of satisfaction of the first condition, suppose that $\tw(G^{\star})>100\Delta^3\sqrt{2k}$. Then, by Proposition~\ref{prop:gridUnitDisk}, $G^{\star}$ contains a $\sqrt{2k}\times\sqrt{2k}$-grid as a minor. Clearly, a $\sqrt{2k}\times\sqrt{2k}$-grid contains a cycle (and hence also a path) on $k$ vertices. By the definition of minor, this means that $G$ contains  cycle (and hence also a path) on at least $k$ vertices, and therefore $(G,k)$ is a \yes-instance\ of {\sc Long Path} and {\sc Long Cycle}. 

Now, we turn to prove the second condition. In one direction, suppose that $(G,k)$ is a \yes-instance of {\sc Long Path} (resp.~{\sc Long Cycle}). Then, by Lemma \ref{lem:propertyMa}, $G$ has a path (resp.~cycle) $P$ on at least $k$ vertices with the following property: every cell $(i,j)\in[t]\times[t]$ is good. Notice that every maximal subpath $Q$ of $P$ that consists only of unmarked vertices satisfies {\em (i)} $V(Q)=f^{-1}(i_Q,j_Q)\setminus \ma^\star(i_Q,j_Q)$ for some cell $(i_Q,j_Q)\in [t]\times[t]$, and {\em (ii)} the endpoints of $Q$ are adjacent in $P$ to vertices in $f^{-1}(i_Q,j_Q)$ (unless $Q=P$). Obtain $P^\star$ from $P$ as follows: every maximal subpath $Q$ of $P$ that consists only of unmarked vertices is replaced by $c_{(i_Q,j_Q)}$. (Notice that  $c_{(i_Q,j_Q)}\neq\nil$ because $V(Q)\neq\emptyset$.) Because of Property {\em (ii)} above and Property \ref{condition:GridClique1} in Definition \ref{def:GridClique}, we immediately have that $P^\star$ is a path (resp.~cycle) in $G^\star$. Moreover, by Property {\em (i)} above and the definition of the weight function $w$ (in Step \ref{alg4}), each subpath $Q$ is replaced by a vertex $c_{(i_Q,j_Q)}$ whose weight equals $|V(Q)|$. Because $|V(P)|\geq k$, we have that $P^\star$ is a path (resp.~cycle) of weight at least $k$ in $G^\star$. Thus, $(G^\star,w,k)$ is a \yes-instance of {\sc Weighted Long Path} (resp.~{\sc Weighted Long Cycle}).

In the other direction, suppose that $(G^\star,w,k)$ is a \yes-instance of {\sc Weighted Long Path} (resp.~{\sc Weighted Long Cycle}). Then, $G^\star$ has a path (resp.~cycle) $P^\star$ of weight at least $k$. 
Obtain $P$ from $P^\star$ by replacing each vertex of the form $c_{(i,j)}\in V(P)$ for some $(i,j)\in[t]\times[t]$ by a path $Q$ whose vertex set is $f^{-1}(i,j)\setminus\ma^\star(i,j)$ (the precise ordering of the vertices on this path is arbitrary). Notice that because all edges in $\{\{c_{(i,j)},v\}\in E(G'): (i,j)\in[t]\times[t], v\notin f^{-1}(i,j)\}$ were removed from $G'$ to derive $G^\star$, each vertex of the form $c_{(i,j)}\in V(P)$ for some $(i,j)\in[t]\times[t]$ is adjacent in $P^\star$ only to vertices in $\ma^\star(i,j)$. Therefore, by Property \ref{condition:GridClique1} in Definition \ref{def:GridClique}, we have that $P$ is a path (resp.~cycle) in $G$. Moreover, by the definition of the weight function $w$ (in Step \ref{alg4}), each vertex $c_{(i,j)}$ was replaced by $w(c_{(i,j)})$ vertices. Because the weight of $P^\star$ is at least $k$, we have that $P$ is a path (resp.~cycle) on at least $k$ vertices in $G$. Thus, $(G,k)$ is a \yes-instance of {\sc Long Path} (resp.~{\sc Long Cycle}).
\end{proof}

Thus, Theorem \ref{thm:main} follows from Lemmas \ref{lem:runtime} and \ref{lem:correctness}.

\bibliography{references}
\end{document}